\documentclass[letterpaper]{sig-alternate-2013}

\newtheorem{theorem}{Theorem}
\newtheorem{lemma}{Lemma}

\newtheorem{definition}{Definition}

\usepackage{ifthen}
\usepackage{epstopdf}
\usepackage[square,comma,numbers]{natbib}
\usepackage{color, xcolor}
\usepackage[caption=false,font=footnotesize]{subfig}
\usepackage{multirow}
\usepackage[flushleft]{threeparttable}
\usepackage{hyperref}
\hypersetup{colorlinks=true,urlcolor=blue,linkcolor=blue,citecolor=blue}

\def\pathfour#1#2#3#4{\overline{#1 #2 #3 #4}}

\newcommand{\compilehidecomments}{false}
\ifthenelse{ \equal{\compilehidecomments}{true} }{%
	\newcommand{\wei}[1]{}
	\newcommand{\yuan}[1]{}
	\newcommend{\del}[1]{}
}{
\newcommand{\wei}[1]{{\color{blue!50!black}  [\text{Wei:} #1]}}
\newcommand{\yuan}[1]{{\color{brown!60!black} [\text{Yuan:} #1]}}
\newcommand{\del}[1]{{\color{red!60!black}
[\textbf{Del:}#1]}}
}

\newcommand{\OnlyInFull}[1]{#1}
\newcommand{\OnlyInShort}[1]{}

\permission{\copyright 2017 International World Wide Web Conference Committee (IW3C2), 
published under Creative Commons CC BY 4.0 License.}
\conferenceinfo{WWW 2017,}{April 3--7, 2017, Perth, Australia.}
\copyrightetc{ACM \the\acmcopyr}
\crdata{978-1-4503-4913-0/17/04. \\
http://dx.doi.org/10.1145/3038912.3052610\\
\includegraphics{cc.png}
}

\makeatletter
\def\@copyrightspace{\relax}
\makeatother

\begin{document}

\sloppy






%

\title{Assessing Percolation Threshold Based on High-Order Non-Backtracking Matrices}



%
%
%
\author{Yuan Lin$^{\dag\ddag}$, Wei Chen$^\S$, and Zhongzhi Zhang$^{\dag\ddag}$ \\
    $^\dag$School of Computer Science, Fudan University, Shanghai, China \\
    $^\ddag$Shanghai Key Laboratory of Intelligent Information Processing, Fudan University, Shanghai, China \\
    $^\S$Microsoft Research, Beijing, China\\
    {yuanlin12@fudan.edu.cn, weic@microsoft.com, zhangzz@fudan.edu.cn}}

\maketitle
\begin{abstract}
Percolation threshold of a network is the critical value such that when nodes or edges are
	randomly selected with probability below the value, the network is fragmented but when the probability
	is above the value, a giant component connecting a large portion of the network would emerge.
Assessing the percolation threshold of networks has wide applications in network reliability, information spread,
	epidemic control, etc.
The theoretical approach so far to assess the percolation threshold is mainly based on spectral radius of
	adjacency matrix or non-backtracking matrix, which is limited to dense graphs
	or locally treelike graphs, and is less effective for sparse networks with non-negligible 
	amount of triangles and loops. 
In this paper, we study high-order non-backtracking matrices and their application to assessing percolation threshold.
We first define high-order non-backtracking matrices and study the properties of their spectral radii.
Then we focus on the 2nd-order non-backtracking matrix and demonstrate analytically that the reciprocal
	of its spectral radius gives a tighter lower
	bound than those of adjacency and standard non-backtracking matrices.
We further build a smaller size matrix with the same largest eigenvalue as the 2nd-order non-backtracking matrix
	to improve computation efficiency.
Finally, we use both synthetic networks and 42 real networks to illustrate that the use of 
	the 2nd-order non-backtracking
	matrix does give better lower bound for assessing percolation threshold than adjacency and standard
	non-backtracking matrices.
%
\end{abstract}

%
%


%
%

%
%



\keywords{percolation theory; percolation threshold; non-backtracking matrix; 
	high-order non-backtracking matrix; information and influence diffusion}

\section{Introduction}

Percolation theory is a powerful statistical physics tool to describe information and viral spreading in social environment~\cite{Wa02, KeKlTa03, JiMiYa14}, robustness and fragility of infrastructural or technological networks~\cite{CaNeSt00,GaScWa11, ChChCh12}, among other things, and thus its impact
	spreads well beyond statistical physics and reaches computer science, network science and
	related areas.
Percolation is a random process independently occupying sites (a.k.a. nodes) or bonds (a.k.a. edges) in a network with probability $p$
	(with probability $1-p$ the site or bond is removed). 
In particular, the bond percolation can be identified to a special case of independent cascade model~\cite{Ne02,KeKlTa03}. 
In this paper, our discussion focuses on bond percolation, although similar approach applies to site percolation as well.
As probability $p$ increases from $0$ to $1$, the network is expected to experience a phase transition from a 
	large number of small connected components to the emergence of a giant connected component, with the size proportional 
	to the size of the network.
The value $p$ at this transition point is referred to as the {\em percolation threshold}.

The percolation threshold can be used to assess the spreading power of a network from a topological 
	point of view.
For example, in studying epidemics in a social network, a small percolation threshold means 
	that a virus is easy to spread in the network
	and infects a large portion of the network.
Thus, many applications rely on the concept of percolation threshold, such as finding influential nodes in a social 
	network~\cite{MoMa15, LeScVa14}, facilitating or curbing propagations~\cite{ChToPr16,ChToPr16TKDD}, and determining transmission rates in
	wireless networking~\cite{FrDoTh07, AnShHe08}.
Therefore, it is crucial to have an accurate understanding of the percolation threshold in networks.

However, since the exact percolation threshold for a general network 
	is analytically difficult to obtain, these applications rely on some theoretical estimates
	of the actual percolation threshold.
A commonly used theoretical estimate is the reciprocal of the largest eigenvalue $\lambda_A$ 
	(same as the spectral radius) of the network's adjacency matrix $A$ (e.g. in~\cite{WaChWa03,BoBoCh10,ChToPr16,ChToPr16TKDD}).
In particular, Bollob\'as \emph{et al}. show that this estimate is accurate for dense networks~\cite{BoBoCh10}.
However, in sparse networks, it could be far off.
For example, in a ring network (Fig.~\ref{fig:ring1}), the real percolation threshold is $1$, but its adjacency matrix has 
	the largest eigenvalue $2$, predicting the percolation threshold is $0.5$.


\begin{figure}[t]
	\centering
	\captionsetup[subfigure]{font=scriptsize,oneside,margin={0.0cm,0.0cm}}
	\setcounter{subfigure}{0}%
	\subfloat[\textrm{ring network}]
	{\includegraphics[width=0.48\linewidth,trim= 50 50 30 0]{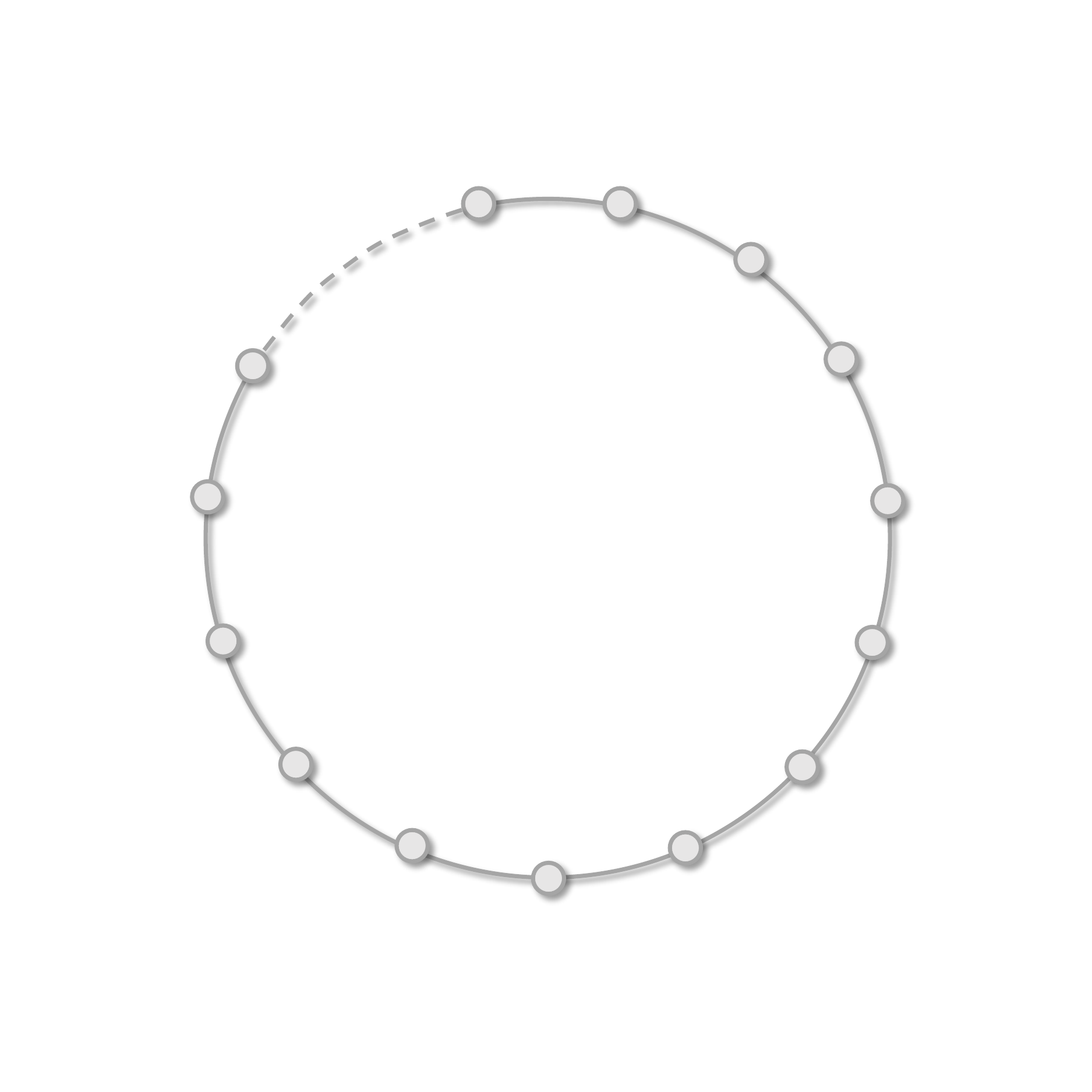}\label{fig:ring1}}
	\subfloat[\textrm{triangle ring network}]
	{\includegraphics[width=0.48\linewidth,trim= 50 50 50 0]{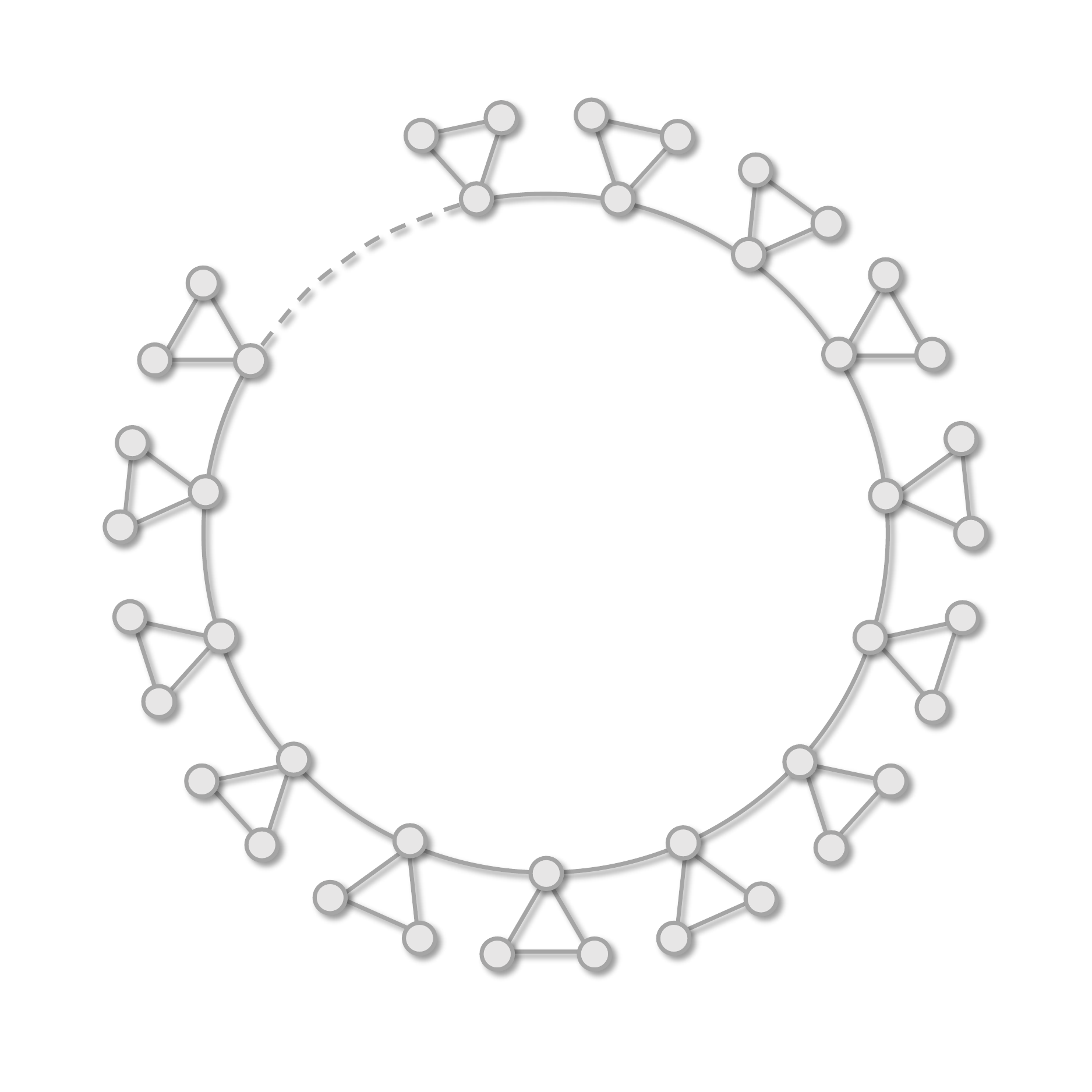}\label{fig:ring2}}
	\caption{Illustrations of the ring network and the triangle ring network.}
	\label{fig:ring}
\end{figure}

Recently, the reciprocal of the largest eigenvalue of the network's non-backtracking matrix~\cite{KaNeZd14,HaPr14} is 
	introduced as a better theoretical estimate of the percolation threshold.
We defer its technical definition to Section~\ref{sec:prelim}, and give
	an intuitive explanation on the issue of the adjacency matrix that is addressed by the non-backtracking matrix.
Let $\pi$ be a vector, with $\pi_i$ representing the probability that node $i$ is in the giant connected component, when every
	edge has a probability of $p$ to be occupied.
Let $A$ be the adjacency matrix of the graph, and ${\cal N}_i$ be the set of neighbors of $i$. The $i$-th entry of $pA\pi$ satisfies $[pA\pi]_i = \sum_{j\in {\cal N}_i} p \pi_j$, which approximately represents that $i$ could connect to the giant component through one of its neighbors $j$ (with probability $p\pi_j$). Thus it should be the same as $\pi_i$, and in matrix form, we have $p A \pi = \pi$.
This suggests that $1/p$ should be at most $\lambda_A$, i.e. $1/\lambda_A$ is a lower bound of the percolation threshold.
However, for every edge $(i,j)$ in the network, the above approximation considers both that $i$ may rely on $j$ to connect to the giant 
	component and $j$ may rely on $i$, but the two cases should not jointly occur.
Therefore, the estimate by $1/\lambda_A$ inflates the probability that a giant component emerges and underestimates the percolation threshold.
The non-backtracking matrix  addresses this issue by disallowing such directed circular dependency between a pair of nodes on an edge.

However, non-backtracking matrix only works well in locally treelike graphs. In non-treelike graphs, it provides a lower bound, bot not as close as the true percolation threshold. As suggested by the famous small-world network study~\cite{WaSt98}, a significant amount of triangles exist in many real networks.
The non-backtracking matrix does not eliminate the circular dependency through triangles or other local structures, 
	and thus it may still underestimate
	the percolation threshold.
For example, in the triangle ring network in Fig.~\ref{fig:ring2}, the real percolation threshold is $1$, but the theoretical
	 estimates by the adjacency and non-backtracking matrices are $0.3333$ and $0.5523$, respectively.

In this paper, we extend the idea of the non-backtracking matrix to define high-order non-backtracking matrices.
We first define the high-order non-backtracking matrices and study the evolution of their largest eigenvalue with respect to order
	(Section~\ref{sec:highNBT}).
Then, we propose that the reciprocal of the largest eigenvalue of the 2nd-order non-backtracking matrix can provide a better estimate for 
	the percolation threshold in an arbitrary network, because it eliminates the circular dependency from triangles (Section~\ref{sec:highNBT}).
We further provide an alternative matrix sharing
	the same largest eigenvalue but with substantially smaller size to improve computation efficiency.
Finally, we conduct extensively experiments on both the forest fire model and 42 real networks to demonstrate the effectiveness of our method
	(Section~\ref{sec:exp}).

To summarize, our contributions include: 
	(i) proposing the high-order non-backtracking matrices and studying
	their eigenvalue properties; 
	(ii) establishing a more precise theoretical estimation for bond percolation threshold and finding a faster approach to evaluate it; and 
	(iii) supporting our analysis by empirical evaluations on synthetic and real networks.


\subsection{Related Work}

For random degree-uncorrelated network models, the percolation threshold can be approximated as $\langle d\rangle/(\langle d^2\rangle-\langle d\rangle)$~\cite{CoErBe00, CaNeSt00}. Here, $\langle d\rangle$ and $\langle d^2\rangle$ are separately the first and the second moments of the degree distribution. This estimation is less predictive for real networks where degree correlations appear.

Bollob\'as \emph{et al}.~\cite{BoBoCh10} show that the percolation threshold of dense graph is reciprocal of the largest eigenvalue of the adjacency matrix. However, the conclusion requires restrictive conditions for the networks. Especially for sparse networks, this estimation can be only regarded as providing a lower bound for the true percolation threshold.

Since a lot of realistic networks are sparse, Karrer \emph{et al}.~\cite{KaNeZd14} and Hamilton \emph{et al}.~\cite{HaPr14} 
	simultaneously propose that the reciprocal of the largest eigenvalue of the non-backtracking matrix is a tighter lower
	bound for bond and site percolation threshold, respectively, on sparse networks. 
This prediction is based on a message passing technique and obtained by heuristic equations and approximations on locally treelike structures.
Radicchi~\cite{Ra15Nat} further presents a mapping between the site and bond percolation to mathematically verify the predicted bond and site percolation thresholds are identical in this method. 
Although the estimation based on the non-backtracking matrix is more precise than that based on the adjacency matrix, 
	it is still not close enough to the true percolation threshold on many real networks~\cite{Ra15}, since it suffers from the limitation of the treelike assumption. 
Radicchi \emph{et al}.~\cite{RaCa16} also derives an alternative matrix of the non-backtracking matrix based on triangle
	elimination to improve the estimate for the
	{\em site} percolation threshold. However, this alternative matrix would overshoot on bond percolation. That means this estimate may overestimate the percolation threshold, leading it no longer provides a lower bound for the bond percolation threshold.

\section{Preliminaries} \label{sec:prelim}

We consider a finite connected undirected graph $\mathcal{G}=(\mathcal{V}, \mathcal{E})$ 
	with $N$ nodes and $E$ edges, where $\mathcal{V}=\{1, 2, \cdots, N\}$ is the set of $N$ nodes and $\mathcal{E}=\{(i, j) \mid i, j\in\mathcal{V}\}$ is the set of $E$ edges. The connectivity between nodes in $\mathcal{G}$ is described by an adjacency matrix ${ A}$, in which the element $a_{ij}=1$ if $(i, j)\in\mathcal{E}$, and $a_{ij}=0$ otherwise. 
We assume that $\cal G$ has no self-loops, i.e. $a_{ii}=0$ for all $i\in {\cal V}$.
Henceforth, the adjacency and non-backtracking matrices are all refer to this graph $\cal G$, when the context is clear.

The bond percolation model is only controlled by one parameter, bond occupation probability $p$. That means, in a network $\mathcal{G}$, each edge is independently occupied with probability $p$. The percolation clusters are sets of nodes connected only by occupied edges. For $p=0$, no edge is occupied so that there are $N$ isolated clusters of size one. For $p=1$, all edges are occupied and all nodes compose a single cluster of size $N$. 
At intermediate values of $p$, the network undergoes two different phases: the non-percolating phase, where all clusters have microscopic size; the percolating phase, where a single macroscopic cluster
	(the giant component), whose size is comparable to the entire network, is present. 
The percolation threshold $p_c$ is the value above which the giant cluster appears; below $p_c$ there are only small clusters.



The transition can be monitored through two primary quantities of interest, the relative sizes of the first and second largest clusters with respect to the size of the network, denoted by $S_1(p)$ and $S_2(p)$, respectively. 
In order to evaluate percolation threshold numerically, there are many different estimates for $p_c$ being proposed, such as the occupation probability corresponding to the maximal value of $S_2(p)$, the peak position of the empirical variance of $S_1(p)$ and the average size of clusters except the largest one~\cite{Ra15}. In this paper, we determine the best empirical estimate for percolation threshold as the value of $p$ where the second largest cluster reaches its maximum, namely
\begin{equation}\label{A0}
p_c=\arg\{\max_p S_2(p)\}.
\end{equation}
In the following, we will compare it with theoretical approaches to check their validity.

The first theoretical estimate is based on the adjacency matrix. Bollob\'as \emph{et al}.~\cite{BoBoCh10} show that for percolation on dense networks, the percolation threshold can be given by
\begin{equation*}
p_c=\frac{1}{\lambda_A},
\end{equation*}
where $\lambda_A$ is the largest eigenvalue of matrix ${ A}$, which is also the
	{\em spectral radius} of $A$.\footnote{In this paper, every matrix we consider is non-negative, and thus
		by the Perron-Frobenius Theorem~\cite{Me00}, its largest eigenvalue is a non-negative real number, in which case
		it is the same as its spectral radius, which is defined
		as the largest modulus of the (possibly complex) eigenvalues.}
Although this estimate provides a good prediction for percolation threshold on networks with high density of connections, it becomes less precise on sparse networks.

Since a lot of realistic networks are sparse, a new estimate based on non-backtracking matrix~\cite{Ha89, KrMoMo13} is proposed to give a better prediction. For any undirected network $\mathcal{G}$, we can transform it to a directed one through replacing each undirected edge $(i, j)$ by two directed ones $i\to j$ and $j\to i$. The non-backtracking matrix, denoted by ${ B}$, is a $2E\times 2E$ matrix with rows and columns indexed by directed edges $i\to j$ and entries are given by
\begin{equation*}\label{A4}
B_{i\to j, k\to l}=\delta_{jk}(1-\delta_{il}),
\end{equation*}
where $\delta_{ij}$ is the Kronecker Delta Function ($\delta_{ij}=1$ if $i = j$, and $0$ if $i\ne j$). 
In other words, for any two directed edges $i \to j, j\to k$,  $B_{i\to j, j\to k} = 1$ if 
	$k\ne i$, and thus $B_{i\to j, j\to k}$ records the walk from $i$ to $j$ and continuing to $k$ but
	not backtracking to $i$.
Under the locally treelike assumption, which means that the local neighborhood of a node is close to a tree
	without redundant paths, the reciprocal of the largest eigenvalue of ${ B}$ provides a close lower bound for sparse networks~\cite{KaNeZd14, HaPr14}, namely
\begin{equation*}\label{A5}
p_c \ge \frac{1}{\lambda_B}, \mbox{and } p_c \approx \frac{1}{\lambda_B}.
\end{equation*}


As already mentioned in the introduction, the non-backtracking matrix relies on the approximation of the locally treelike structure and
	it does not eliminate triangle dependency, and thus it may not give a tight lower bound of the bond percolation threshold.
In the following, we propose a more powerful tool, {\em high-order non-backtracking matrices}, to better estimate bond percolation threshold.


\section{High-Order Non-Backtracking \\ Matrices} \label{sec:highNBT}

In this section, we first introduce the definition of high-order non-backtracking matrices, 
	and then investigate their properties.


In a network $\mathcal{G}$, let $i_1\to i_2\to\cdots \to i_{g+1}$ stand for a {\em length-$g$ directed path}
	composed by $g+1$ {\em different} nodes $i_1, i_2, \cdots, i_{g+1}$, obeying $a_{i_{k}i_{k+1}}=1$ ($1\leq k\leq g$). For the sake of saving space, we use $\pathfour{i_1}{i_2}{\cdots}{i_{g+1}}$ to stand for length-$g$ directed path in place of $i_1\to i_2\to\cdots\to i_{g+1}$. Let $P_g$ be the number of length-$g$ directed paths in $\mathcal{G}$. 
Then, we define the {\em $g$-th-order non-backtracking matrix} as follows:

\begin{definition}
The $g$-th-order non-backtracking matrix ${ B}^{(g)}$ is a $P_g\times P_g$ matrix with rows and columns indexed by length-$g$ directed paths. The elements in ${ B}^{(g)}$ are
\begin{equation} \label{eq:HONBMdef}
B^{(g)}_{\pathfour{i_1}{i_2}{\cdots}{i_{g+1}}, \overline{j_1 j_2\cdots j_{g+1}}}=
	a_{i_1 j_1}(1-\delta_{i_1 j_{g+1}})\prod_{k=2}^{g+1}\delta_{i_k j_{k-1}}.
\end{equation}

\end{definition}
In other words, for each length-$(g+1)$ directed path, $\pathfour{i_1}{i_2}{\cdots}{i_{g+2}}$, we have $B^{(g)}_{\pathfour{i_1}{i_2}{\cdots}{i_{g+1}}, \pathfour{i_2}{i_3}{\cdots}{i_{g+2}}} = 1$; and all other entries of $B^{(g)}$ have value $0$.
For example, if $\pathfour{1}{2}{3}{}$, $\pathfour{2}{3}{4}{}$, $\pathfour{3}{4}{5}{}$, and $\pathfour{2}{3}{1}{}$
	are all length-$2$ directed paths in a graph,  then 
	${B}^{(2)}_{\pathfour{1}{2}{3}{},\pathfour{2}{3}{4}{}}=1$, while
	${B}^{(2)}_{\pathfour{1}{2}{3}{},\pathfour{3}{4}{5}{}}= {B}^{(2)}_{\pathfour{1}{2}{3}{},\pathfour{2}{3}{1}{}} = 0$.
We can regard matrix ${B}^{(g)}$ as encoding the relations between length-$g$ directed paths in $\mathcal{G}$. It also describes a kind of non-backtracking walks with memory, avoiding going back to a node visited in the recent $g$ steps. 
It is easy to verify that the standard non-backtracking matrix is $B={B}^{(1)}$, 
	and adjacency matrix $A = {B}^{(0)}$ ($a_{i_1 j_1}$ in Eq.~\eqref{eq:HONBMdef} is explicitly for
	the case of $g=0$). 


Let $\lambda_B^{(g)}$ be the largest eigenvalue of ${B}^{(g)}$. 
By the Perron-Frobenius Theorem~\cite{Me00}, $\lambda_B^{(g)}$ is real and non-negative, and it is zero only when
	the directed graph with the adjacency matrix equal to ${B}^{(g)}$ is a directed acyclic graph (DAG).
We now study the properties of $\lambda_B^{(g)}$ as $g$ changes.
A {\em simple cycle} in graph $\mathcal{G}$ is a closed path with no repetitions of nodes and directed edges allowed, except the repetition of the starting and ending nodes. To maintain a consistent terminology, henceforth, all {\em cycle} means a simple cycle. The result is given in the following theorem:

\begin{theorem}\label{T0}
The $g$-th-order non-backtracking matrix ${B}^{(g)}$ satisfies the following properties:

\emph{(i)} The largest eigenvalue of ${B}^{(g)}$ is non-increasing with respect to $g$, i.e., for every $g=1, 2, 3, \cdots$, $\lambda_B^{(g-1)}\geq\lambda_B^{(g)}$.

\emph{(ii)} Let $\ell$ be the length of the longest simple cycle in graph $\mathcal{G}$. 
Then $\lambda_B^{(g)} = 0$ for all $g\geq \ell-1$.

\emph{(iii)} For any $\ell\ge 3$, if there is no simple cycle with length $\ell$ in $\cal G$, 
then $\lambda_B^{(\ell-2)} = \lambda_B^{(\ell-1)}$.
\end{theorem}


In order to prove Theorem~\ref{T0}, we first present some related notations and lemmata. 
The $g$-th-order non-backtracking matrix ${B}^{(g)}$ of network $\mathcal{G}$ can be 
	regarded as the adjacency matrix of a directed network $\mathcal{G}^{(g)}=(\mathcal{V}^{(g)}, \mathcal{E}^{(g)})$, where $\mathcal{V}^{(g)}=\{u\mid u $ is a length-$g$ directed path$\}$, and $\mathcal{E}^{(g)}=\{u \to v \mid u, v$ are two length-$g$ directed paths and $B^{(g)}_{u, v}=1\}$. 
Note that when $g=0$, $\mathcal{G}^{(0)}=\mathcal{G}$ (considering $\mathcal{G}$ as a symmetric
	directed graph). 
The line graph of $\mathcal{G}^{(g)}$ is $\bar{\mathcal{G}}^{(g)}=(\bar{\mathcal{V}}^{(g)}, \bar{\mathcal{E}}^{(g)})$, where $\bar{\mathcal{V}}^{(g)}=\mathcal{E}^{(g)}$ and $\bar{\mathcal{E}}^{(g)}=\{\alpha \to \beta \mid \alpha=(u, v)$ and $\beta=(v, w)$ are two different edges $u\to v$ and $v\to w$ in $\mathcal{E}^{(g)}.\}$. Let $\bar{B}^{(g)}$ denote the adjacency matrix of the line graph $\bar{\mathcal{G}}^{(g)}$, with elements
\begin{eqnarray*}
\bar{B}^{(g)}_{u\to v, p\to q}
=\delta_{vp},
\end{eqnarray*}
where $u$, $v$, $p$ and $q$ are length-$g$ directed paths with
	$u\to v, p \to q \in {\cal E}^{(g)}$. In other words, we can recast definition of $\bar{B}^{(g)}$ as 
\begin{eqnarray*}
\bar{B}^{(g)}_{\pathfour{i_1}{i_2}{\cdots}{i_{g+2}}, \pathfour{j_1}{j_2}{\cdots}{j_{g+2}}}=\prod_{k=2}^{g+2}\delta_{i_k j_{k-1}}.
\end{eqnarray*}

\begin{lemma}\label{Lem2}
$\mathcal{G}^{(g)}$ can be obtained by deleting all edges in all length-$(g+1)$ cycles in $\bar{\mathcal{G}}^{(g-1)}$.
\end{lemma}
\begin{proof}
For each edge $u\to v$ ($u=\pathfour{i_1}{i_2}{\cdots}{i_g}$ and $v=\pathfour{i_2}{\cdots}{i_g}{i_{g+1}}$) in $\mathcal{G}^{(g-1)}$, it must correspond to a length-$g$ directed path $\pathfour{i_1}{\cdots}{i_g}{i_{g+1}}$ in $\mathcal{G}$. On the other hand, each length-$g$ directed path also corresponds to an edge in $\mathcal{G}^{(g-1)}$. 
Thus, $\mathcal{V}^{(g)}=\mathcal{E}^{(g-1)}=\bar{\mathcal{V}}^{(g-1)}$
	if we identify $\pathfour{i_1}{i_2}{\cdots}{i_g} \to \pathfour{i_2}{\cdots}{i_g}{i_{g+1}}$
	with $\pathfour{i_1}{\cdots}{i_g}{i_{g+1}}$,  i.e., 
the node set of $\bar{\mathcal{G}}^{(g-1)}$ is identical to
	the node set of $\mathcal{G}^{(g)}$. 
Moreover, for every edge in $\mathcal{G}^{(g)}$, $\pathfour{i_1}{i_2}{\cdots}{i_{g+1}}\to\pathfour{i_2}{\cdots}{i_{g+1}}{i_{g+2}}{}$, 
	it is also an edge in $\bar{\mathcal{G}}^{(g-1)}$, implying $\mathcal{E}^{(g)}\subseteq\bar{\mathcal{E}}^{(g-1)}$. 
For an arbitrary edge in $\bar{\mathcal{E}}^{(g-1)} \setminus \mathcal{E}^{(g)}$, it must have a form as $\pathfour{i_1}{i_2}{\cdots}{i_{g+1}}\to\pathfour{i_2}{\cdots}{i_{g+1}}{i_1}$, and there must be $g$ other edges with the same form together composing a length-$(g+1)$ cycle in $\bar{\mathcal{G}}^{(g-1)}$, namely $\pathfour{i_1}{i_2}{\cdots}{i_{g+1}}\to\pathfour{i_2}{\cdots}{i_{g+1}}{i_1}\to\cdots\to\pathfour{i_{g+1}}{i_1}{\cdots}{i_{g}}\to\pathfour{i_1}{i_2}{\cdots}{i_{g+1}}$. 
Moreover, for any length-$(g+1)$ cycle in $\bar{\mathcal{G}}^{(g-1)}$, it must be of the above form.
If not, then it involves an edge of form $\pathfour{i_1}{i_2}{\cdots}{i_{g+1}}\to\pathfour{i_2}{\cdots}{i_{g+1}}{i_{g+2}}$ with
		$i_{g+2}\ne i_1$, the cycle must be of length at least $g+2$, because we have $g+2$ different original nodes
		$i_1, \ldots, i_{g+2}$ in the cycle, and each edge in the cycle only removes one head node and adds one tail node, and thus
		it needs at least $g+2$ edges to go through every original node once and comes back to $\pathfour{i_1}{i_2}{\cdots}{i_{g+1}}$.
Together, we know that by exactly removing all edges in all length-$(g+1)$ cycles in $\bar{\mathcal{G}}^{(g-1)}$,
	we obtain $\mathcal{G}^{(g)}$.
\end{proof}

The next lemma is a known result saying that when we remove edges from a graph, the largest eigenvalue
	of the adjacency matrix will not increase.
\begin{lemma}[Proposition 3.1.1 of \cite{BrHa11}]\label{LEM4}
For the adjacency matrix $A$ of a graph $\mathcal{G}$, if $A'$ is a matrix obtained by replacing some
	of the $1$'s in $A$ with $0$, then we have
\begin{equation}
\lambda_{A'}\leq\lambda_A,
\end{equation}
where $\lambda_{A'}$ and $\lambda_A$ are the largest eigenvalue of $A'$ and $A$, respectively.
\end{lemma}

\OnlyInFull{
For convenience, we give an independent proof of Lemma~\ref{LEM4} in Appendix~\ref{App1}.
}

We are now ready to prove Theorem~\ref{T0}.



\begin{proof} [of Theorem~\ref{T0}]$ $

(i) For a directed graph $\mathcal{G}$ and its line graph, Pako{\'n}ski et al.~\cite{PaTaZy03} show that 
	their spectra excluding the possible eigenvalue zero are exactly the same.
Therefore, we have
\begin{equation*}
\lambda_B^{(g-1)}=\bar{\lambda}_B^{(g-1)}, 
\end{equation*}
where $\bar{\lambda}_B^{(g-1)}$ is the largest eigenvalue of matrix $\bar{B}^{(g-1)}$. 

According to Lemma~\ref{Lem2}, $\mathcal{G}^{(g)}$ have the same node set as $\bar{\mathcal{G}}^{(g-1)}$, and $\mathcal{G}^{(g)}$ can be obtained by removing edges in length-$(g+1)$ cycles in $\bar{\mathcal{G}}^{(g-1)}$. Thus, applying Lemma~\ref{LEM4}, we attain the conclusion that
\begin{equation*}
\lambda_B^{(g-1)}=\bar{\lambda}_B^{(g-1)}\geq \lambda_B^{(g)}.
\end{equation*}

(ii) 
We first claim that, $\mathcal{G}^{(g)}$ is a DAG for all $g\geq \ell-1$. 
If not, there are $k$ ($k\geq g+2$) length-$g$ directed paths in ${\cal G}^{(g)}$, 
	$\pathfour{i_1}{i_2}{\cdots}{i_{g+1}}$, $\pathfour{i_2}{i_3}{\cdots}{i_{g+2}}$, \dots, $\pathfour{i_k}{i_1}{\cdots}{i_g}$, constituting a cycle in $\mathcal{G}^{(g)}$.
Then there must be a length-$k$ cycle composed by $i_1, i_2, \cdots,i_{k}$ in $\mathcal{G}$. 
Since $g\geq \ell-1$, we have $k\geq\ell+1$, implying that there is at least one cycle in $\cal G$
	with length larger than $\ell$, which contradicts to the assumption that the length of the largest cycle in $\mathcal{G}$ is $\ell$. 
Thus, network $\mathcal{G}^{(g)}$ is a DAG.
It is easy to verify that the adjacency matrix of a DAG has the normal form where
	all block matrices on the diagonal are one-dimensional matrices with value $0$, and thus
	all its eigenvalues are $0$.
Therefore $\lambda_B^{(g)} = 0$ for all $g\geq \ell-1$.

(iii) Consider the line graph of ${\cal G}^{(\ell-2)}$, $\bar{\cal G}^{(\ell-2)}$.
By Lemma~\ref{Lem2}, graph ${\cal G}^{(\ell-1)}$ is obtained by removing all
	length-$\ell$ cycles in $\bar{\cal G}^{(\ell-2)}$.
Similar to the argument in the proof of  Lemma~\ref{Lem2}, every length-$\ell$ cycle
	in $\bar{\cal G}^{(\ell-2)}$ has the form 
	$\pathfour{i_1}{i_2}{\cdots}{i_\ell}$ $\to$ $\pathfour{i_2}{\cdots}{i_\ell}{i_1} \to$ 
	$\cdots \pathfour{i_\ell}{i_1}{\cdots}{i_{\ell-1}}$ $\to \pathfour{i_1}{i_2}{\cdots}{i_\ell}$.
Then every such cycle corresponds to a length-$\ell$ simple cycle in the original graph $\cal G$,
	$i_1 \to i_2 \to \cdots i_\ell \to i_1$.
By assumption $\cal G$ has no length-$\ell$ simple cycles, thus it implies that
	$\bar{\cal G}^{(\ell-2)}$ is the same as ${\cal G}^{(\ell-1)}$.
Since the line graph shares the same non-zero eigenvalues as the original graph~\cite{PaTaZy03},
	we know that $\lambda^{(\ell-1)} = \lambda^{(\ell-2)}$.
\end{proof}

\section{Estimating Percolation Threshold by 2nd-Order Non-Backtracking Matrix} \label{sec:2ndNBT}



In this section, we show that the reciprocal of the 
	largest eigenvalue of the 2nd-order non-backtracking matrix
	gives a lower bound on the bond percolation threshold.
According to Theorem~\ref{T0}, it implies that this lower bound is tighter than the previous proposed
	analytical lower bounds using the reciprocal of the adjacency matrix 
	or the standard non-backtracking matrix.
We then show that we can replace the 2nd-order non-backtracking matrix with a substantially smaller matrix
	sharing the same largest eigenvalue to improve computation efficiency.
	
\subsection{Derivation of Lower Bound} \label{sec:lowerbound}
	
Our derivation follows the message passing techniques proposed in \cite{KaNeZd14, HaPr14}, 
	which involves first-order analysis (ignoring higher-order terms) and 
	heuristic equations on locally treelike graphs.
Henceforth, we denote percolation threshold predicted by the 
	adjacency matrix, non-backtracking matrix and 2nd-order non-backtracking matrix as 
	$p_c^{(0)}$, $p_c^{(1)}$ and $p_c^{(2)}$, respectively.


Let $\pi_i$ be the probability that node $i$ belongs to the giant connected component, where the
	random events are that each edge in $\cal G$ is independently occupied with probability $p$.\footnote{Technically,
		a giant connected component is a component of size $\Theta(n)$, where $n$ is the size of the graph, and thus
		$\Theta(n)$ is only meaningful when we have a series of graphs with $n$ goes to infinity. The argument in this section
		does not take this techical route, and instead it follows the heuristic argument approach as in~\cite{KaNeZd14, HaPr14}.}
Let $\theta_{i\to j\to k}$ be the probability that node $i$ connects to the giant component following path $i\to j\to k$, where
	$i,j,k$ are all different nodes. 
Then, we can construct the following relation for bond percolation:
\begin{equation}\label{G1}
\pi_i=1-\prod_{\substack{j\to k\\ j\in\mathcal{N}_i, k\neq i}}(1-\theta_{i\to j\to k}),
\end{equation}
Then, the expected size of the giant component can be given by
\begin{equation*}\label{G2}
S_1(p)=\sum_{i=1}^{N}\pi_i.
\end{equation*}

In a finite-size network, there is a drastic change for $S_1(p)$ at the percolation threshold $p_c$. Next, we focus on predicting 
	the value of $p_c$.

According to the definition of quantity $\theta_{i\to j \to k}$, we can construct recursive relations, which is appropriate for
	locally treelike structures with triangles:
\begin{equation}\label{B5}
\theta_{i\to j\to k}=1-\prod_{\ell\in\mathcal{N}_k\setminus \{i,j\}}(1-p \cdot \theta_{j\to k\to \ell}).
\end{equation}
The above heuristic equation intuitively means that, node $i$ connects to the giant component
	through at least one of the paths $j\to k\to \ell$, and $i$ connects to $j$ with probability $p$
	while $j$ connects to the giant component through path $j\to k\to \ell$ with probability
	$\theta_{j\to k\to \ell}$.
The equation is approximately accurate when the local neighborhood of $i$ is close to a tree without
	redundant paths, except that we allow triangles such as $i\to j\to k \to i$, since the 
	equation requires $\ell \ne i$, excluding such triangles.
When we ignore $p^2$ and higher order terms in Eq.~\eqref{B5}, we obtain
%
\begin{eqnarray*}\label{B7}
\theta_{i\to j\to k}=p\sum_{\ell\in\mathcal{N}_k\setminus\{i, j\}}\theta_{j\to k\to \ell},
\end{eqnarray*}
which can be recast in matrix notation as
\begin{eqnarray*}\label{B8}
\theta=p{B}^{(2)}\theta,
\end{eqnarray*}
where $\theta$ is a vector in which elements indexed by length-2 directed paths. 
When $p\ge p_c$, percolation happens, and $\theta$ should be a non-negative vector with some strictly positive entries. (It holds when length of the largest cycle is at least 4.) But if $p_c < 1/\lambda^{(2)}_B$, it means we could have an eigenvalue of $1/p_c > \lambda^{(2)}_B$, violating the definition of
	$\lambda^{(2)}_B$.
Therefore, we have 
%
\begin{eqnarray*}\label{B9}
p_c \ge p_c^{(2)}=\frac{1}{\lambda^{(2)}_B}. 
\end{eqnarray*}




Theorem~\ref{T0} theoretically guarantees $p_c^{(2)}\geq p_c^{(1)} \geq p_c^{(0)} $, 
	so that $p_c^{(2)}$ provides a tighter lower bound for percolation threshold than the ones
	provided by the adjacency and standard non-backtracking matrices. 
Moreover, consider the example of triangle ring network shown in Figure~\ref{fig:ring2}, it is easy to obtain $\lambda_B^{(2)}=1$ and $p_c^{(2)}=1/\lambda_B^{(2)}=1$, which is consistent with the true percolation threshold. 
Compared with estimates $p_c^{(0)}$ and $p_c^{(1)}$, $p_c^{(2)}$ exhibits remarkable improvement in precision. 

The reason why $p_c^{(2)}$ provides better prediction of $p_c$ than $p_c^{(1)}$ can be heuristically explained as follows. 
In the estimation based on non-backtracking matrix, it only considers the probability node $i$ connects to the giant component through node $j$, which can be denoted as $\theta_{i\to j}$. 
If there is a triangle composed by nodes $i$, $j$ and $k$ in the network, 
	$\theta_{i\to j}$ grows as $\theta_{j\to k}$ increases. 
Similarly, $\theta_{j\to k}$ increases with $\theta_{k\to i}$, and
	$\theta_{k\to i}$ increases with $\theta_{i\to j}$. 
This creates a triangular dependency which artificially inflates the values of
	$\theta_{i\to j}$, $\theta_{j\to k}$, and $\theta_{k\to i}$, leading to a higher estimate
	of the probability of giant component emergence and a lower estimate on the percolation threshold.
When triangles are abundant, as evidenced by the small-world research on 
	many real-world networks~\cite{WaSt98}, using the standard non-backtracking matrix may still
	significantly underestimate the percolation threshold.
By using the 2nd-order non-backtracking matrix, we avoid the
	triangular  dependency, so that $p_c^{(2)}$ is more precise than $p_c^{(1)}$.


Although $p_c^{(2)}$ is a tighter lower bound, the size of ${B}^{(2)}$ is usually larger than ${B}$, leading to higher computation complexity. 
In the following we provide a further technique to tackle this problem.

\subsection{Improving Computation Efficiency}

In this subsection, we illustrate that we can transfer the task of computing eigenvalues of ${B}^{(2)}$ to calculating eigenvalues of a new matrix ${M}$, defined as
\begin{eqnarray}\label{D1}
{M}=\left(\begin{array}{cccc}
{B}^{(1)} & -\Delta{B}_2 & {D}_{\Delta}-{I} & {B}^{(1)} - \Delta{B}_1 \\
{I} & {0} & {0} & {0} \\
{0} & {I} & {0} & {0} \\
{0} & {0} & {I} & {0}\end{array}
\right).
\end{eqnarray}
Here, ${M}$ is a square matrix of dimension $8E\times 8E$ composed of $16$ $2E\times 2E$ blocks, whose size is evidently smaller than ${B}^{(2)}$. Matrix $\Delta{B}_1$ and $\Delta{B}_2$ have identical size with ${B}^{(1)}$ and their elements are defined as
\begin{eqnarray*}\label{D2}
(\Delta B_1)_{i\to j, k\to \ell}=\left\{\begin{aligned}
&1, \quad a_{i\ell}B^{(1)}_{i\to j, k\to \ell}=1,\\
&0, \quad {\rm otherwise.}
\end{aligned}
\right.
\end{eqnarray*}
and
\begin{eqnarray*}\label{D3}
(\Delta B_2)_{i\to j, k\to \ell}=\left\{\begin{aligned}
&1, \quad \delta_{i\ell}a_{ij}a_{jk}a_{k\ell}=1,\\
&0, \quad {\rm otherwise,}
\end{aligned}
\right.
\end{eqnarray*}
respectively. The matrix ${D}_{\Delta}$ is a $2E\times 2E$ diagonal matrix and its element $(D_{\Delta})_{i\to j, i\to j}$ equals the number of triangles containing edge $(i,j)$.

\begin{theorem}\label{T2}
The set of non-zero eigenvalues of matrix $M$ defined in Eq.~(\ref{D1}) consists of
	all non-zero eigenvalues of the 2nd-order non-backtracking matrix ${B}^{(2)}$ 
	and possibly eigenvalues $-1$, $(1+\sqrt{3}\mathrm{i})/2$ and $(1-\sqrt{3}\mathrm{i})/2$. 
\end{theorem}

\begin{proof}
The proof will include two parts. First, we prove that every non-zero eigenvalue of $B^{(2)}$ corresponds to an eigenvalue of $M$. Second, we elaborate that every non-zero eigenvalue of $M$ equals to one of eigenvalues of $B^{(2)}$, except possible eigenvalues $-1$, $(1+\sqrt{3}\mathrm{i})/2$ and $(1-\sqrt{3}\mathrm{i})/2$.

(i) Let $\lambda$ be an arbitrary non-zero eigenvalue of ${B}^{(2)}$ and $\psi$ be its corresponding eigenvector. According to the definition of the 2nd-order non-backtracking matrix, we have
\begin{eqnarray}\label{F3}
\lambda \psi_{i\to j\to k}=\sum_{\ell\in\mathcal{N}_k\setminus\{i, j\}}B^{(1)}_{j\to k, k\to l}\psi_{j\to k\to \ell}.
\end{eqnarray}

In order to show each non-zero eigenvalue of $B^{(2)}$ is also an eigenvalue of $M$, we first introduce two related quantities:
\begin{equation}\label{F2}
x_{i\to j}=\sum_{k\in\mathcal{N}_j\setminus\{i\}}B^{(1)}_{i\to j, j\to k}\psi_{i\to j\to k}
\end{equation}
and
\begin{equation}\label{F20}
y_{i\to j}=\sum_{k\in\mathcal{N}_j\setminus\{i\}}B^{(1)}_{i\to j, j\to k}B^{(1)}_{j\to k, k\to i}\psi_{i\to j\to k}.
\end{equation}
Let $x$ and $y$ be vectors with elements $x_{i\to j}$ and $y_{i\to j}$, respectively. We claim that $x$ and $y$ satisfies the following two relations:
\begin{equation}\label{F8}
x=\frac{1}{\lambda}{B}^{(1)}x-\frac{1}{\lambda^2}\Delta{B}_2x+\frac{1}{\lambda^3}{D}_{\Delta}x-\frac{1}{\lambda^3}y
\end{equation}
and
\begin{eqnarray}\label{F13}
y=\frac{1}{\lambda}\Delta{B}_1 x - \frac{1}{\lambda ^2}\Delta{B}_2 x + \frac{1}{\lambda^3}{D}_\Delta x-\frac{1}{\lambda^3}y,
\end{eqnarray}
which will be proved latter. 

Combining Eq.~(\ref{F8}) and Eq.~(\ref{F13}), we can obtain
\begin{eqnarray*}\label{F14}
x-y=\frac{1}{\lambda}({B}^{(1)}-\Delta{B}_1)x,
\end{eqnarray*}
which indicates
\begin{eqnarray}\label{F15}
y=\left({I}-\frac{1}{\lambda}{B}^{(1)}+\frac{1}{\lambda}\Delta{B}_1\right)x.
\end{eqnarray}
Plugging Eq.~(\ref{F15}) into Eq.~(\ref{F8}) yields
\begin{eqnarray}\label{F16}
x&=&\frac{1}{\lambda}{B}^{(1)}x-\frac{1}{\lambda^2}\Delta{B}_2 x + \frac{1}{\lambda^3}\left({D}_{\Delta}-{I}\right)x\nonumber\\
&\quad&+\frac{1}{\lambda^4}({B}^{(1)}-\Delta{B}_1)x.
\end{eqnarray}
Based on Eq.~(\ref{F16}), we establish
\begin{eqnarray*}\label{F17}
{M}z=\lambda z,
\end{eqnarray*}
where $z=(x^\top,\frac{1}{\lambda}x^\top,\frac{1}{\lambda^2}x^\top,\frac{1}{\lambda^3}x^\top)^\top$. Thus, every non-zero eigenvalue of $B^{(2)}$ corresponds to an eigenvalue of $M$.

\emph{Proof of} Eq.~(\ref{F8}). Plugging Eq.~(\ref{F3}) into Eq.~(\ref{F2}), we have
\begin{eqnarray}\label{F4}
x_{i\to j}&=&\frac{1}{\lambda}\sum_{k\in\mathcal{N}_j\setminus\{i\}}B^{(1)}_{i\to j,j\to k}\sum_{\ell\in\mathcal{N}_k\setminus\{j\}} B^{(1)}_{j\to k, k\to \ell}\psi_{j\to k\to \ell}\nonumber\\
&\quad&-\frac{1}{\lambda}\sum_{k\in\mathcal{N}_j \cap \mathcal{N}_i}B^{(1)}_{i\to j, j\to k}B^{(1)}_{j\to k, k\to i}\psi_{j\to k\to i}.
\end{eqnarray}
In the r.h.s of Eq.~(\ref{F4}), the first term can be rewritten as
\begin{eqnarray}\label{F5}
&\quad&\frac{1}{\lambda}\sum_{k\in\mathcal{N}_j\setminus\{i\}}B^{(1)}_{i\to j,j\to k}\sum_{\ell\in\mathcal{N}_k\setminus\{j\}} B^{(1)}_{j\to k, k\to \ell}\psi_{j\to k\to \ell}\nonumber\\
&=&\frac{1}{\lambda}\sum_{k\in\mathcal{N}_j\setminus\{i\}}B^{(1)}_{i\to j, j\to k}x_{j\to k}.
\end{eqnarray}
Here, Eq.~(\ref{F2}) is used. Utilizing Eq.~(\ref{F3}) again, the second term of r.h.s of Eq.~(\ref{F4}) can be represented as
\begin{eqnarray}\label{F6}
&\quad&\frac{1}{\lambda}\sum_{k\in\mathcal{N}_j\cap\mathcal{N}_i}B^{(1)}_{i\to j, j\to k}B^{(1)}_{j\to k, k\to i}\psi_{j\to k\to i}\nonumber\\
&=&\frac{1}{\lambda^2}\sum_{k\in\mathcal{N}_j\cap\mathcal{N}_i}B^{(1)}_{i\to j, j\to k}B^{(1)}_{j\to k, k\to i}\sum_{\ell\in\mathcal{N}_i\setminus\{k\}}B^{(1)}_{k\to i, i\to \ell}\psi_{k\to i\to \ell}\nonumber\\
&\quad&-\frac{1}{\lambda^2}\sum_{k\in\mathcal{N}_j\cap\mathcal{N}_i}B^{(1)}_{i\to j, j\to k}B^{(1)}_{j\to k, k\to i}B^{(1)}_{k\to i, i\to j}\psi_{k\to i\to j}\nonumber\\
&=&\frac{1}{\lambda^2}\sum_{k\in\mathcal{N}_j\cap\mathcal{N}_i}B^{(1)}_{i\to j, j\to k}B^{(1)}_{j\to k, k\to i}x_{k\to i}\nonumber\\
&\quad&-\frac{1}{\lambda^3}\sum_{k\in\mathcal{N}_j\cap\mathcal{N}_i}B^{(1)}_{i\to j, j\to k}B^{(1)}_{j\to k, k\to i}\sum_{\ell\in\mathcal{N}_j\setminus\{i\}}B^{(1)}_{i\to j, j\to \ell}\psi_{i\to j\to \ell}\nonumber\\
&\quad&+\frac{1}{\lambda^3}\sum_{k\in\mathcal{N}_j\cap\mathcal{N}_i}B^{(1)}_{i\to j, j\to k}B^{(1)}_{j\to k, k\to i}\psi_{i\to j\to k}\nonumber\\
&=&\frac{1}{\lambda^2}\sum_{k\in\mathcal{N}_j\cap\mathcal{N}_i}B^{(1)}_{i\to j, j\to k}B^{(1)}_{j\to k, k\to i}x_{k\to i}\nonumber\\
&\quad&-\frac{\Delta_{i\to j}}{\lambda^3}x_{i\to j}+\frac{1}{\lambda^3}y_{i\to j},
\end{eqnarray}
where $\Delta_{i\to j}$ is the number of length-3 loops starting from edge $i\to j$.

Inserting Eq.~(\ref{F5}) and Eq.~(\ref{F6}) into Eq.~(\ref{F4}), we have
\begin{eqnarray*}\label{F7}
x_{i\to j}&=&\frac{1}{\lambda}\sum_{k\in\mathcal{N}_j\setminus\{i\}}B^{(1)}_{i\to j, j\to k}x_{j\to k}\nonumber\\
&\quad&-\frac{1}{\lambda^2}\sum_{k\in\mathcal{N}_j\cap\mathcal{N}_i}B^{(1)}_{i\to j, j\to k}B^{(1)}_{j\to k, k\to i}x_{k\to i}\nonumber\\
&\quad&+\frac{\Delta_{i\to j}}{\lambda^3}x_{i\to j}-\frac{1}{\lambda^3}y_{i\to j}.
\end{eqnarray*}
Recasting Eq.~(\ref{F7}) in matrix notation, we can obtain Eq.~(\ref{F8}).


\emph{Proof of} Eq.~(\ref{F13}). Analogously, substituting Eq.~(\ref{F3}) into Eq.~(\ref{F20}), $y_{i\to j}$ can be represented as
\begin{eqnarray}\label{F9}
y_{i\to j}&=&\frac{1}{\lambda}\sum_{k\in\mathcal{N}_j\setminus\{i\}}B^{(1)}_{i\to j, j\to k}B^{(1)}_{j\to k, k\to i}\nonumber\\
&\quad&\sum_{\ell\in\mathcal{N}_k\setminus\{j\}}B^{(1)}_{j\to k, k\to \ell}\psi_{j\to k\to \ell}\nonumber\\
&\quad&-\frac{1}{\lambda}\sum_{k\in\mathcal{N}_j\setminus\{i\}}B^{(1)}_{i\to j, j\to k}B^{(1)}_{j\to k, k\to i}\psi_{j\to k\to i}.
\end{eqnarray}
The first term in r.h.s of Eq.~(\ref{F9}) can be further expressed as
\begin{eqnarray}\label{F11}
&\quad&\frac{1}{\lambda}\sum_{k\in\mathcal{N}_j\setminus\{i\}}B^{(1)}_{i\to j, j\to k}B^{(1)}_{j\to k, k\to i}\sum_{\ell\in\mathcal{N}_k\setminus\{j\}}B^{(1)}_{j\to k, k\to \ell}\psi_{j\to k\to \ell}\nonumber\\
&=&\frac{1}{\lambda}\sum_{k\in\mathcal{N}_j\setminus\{i\}}B^{(1)}_{i\to j, j\to k}B^{(1)}_{j\to k, k\to i}x_{j\to k}.
\end{eqnarray}
Instituting Eq.~(\ref{F6}) and Eq.~(\ref{F11}) into Eq.~(\ref{F9}), one has
\begin{eqnarray*}\label{F12}
y_{i\to j}&=&\frac{1}{\lambda}\sum_{k\in\mathcal{N}_j\setminus\{i\}}B^{(1)}_{i\to j, j\to k}B^{(1)}_{j\to k, k\to i}x_{j\to k}\nonumber\\
&\quad&-\frac{1}{\lambda^2}\sum_{k\in\mathcal{N}_j\setminus\{i\}}B^{(1)}_{i\to j,j\to k}B^{(1)}_{j\to k, k\to i}x_{k\to i}\nonumber\\
&\quad&+\frac{\Delta_{i\to j}}{\lambda^3}x_{i\to j}-\frac{1}{\lambda^3}y_{i\to j},
\end{eqnarray*}
which can be recast in matrix form to obtain Eq.~(\ref{F13}).


(ii) Now, we prove that every non-zero eigenvalue of $M$ equals to one of eigenvalues of $B^{(2)}$, 
	except possible eigenvalues $-1$, $(1+\sqrt{3}\mathrm{i})/2$ and $(1-\sqrt{3}\mathrm{i})/2$. Let $\lambda$ be an arbitrary non-zero eigenvalue of $M$ and $z=(x^\top,\frac{1}{\lambda}x^\top,\frac{1}{\lambda^2}x^\top,\frac{1}{\lambda^3}x^\top)^\top$ be its corresponding eigenvector. According to Eq.~(\ref{F3}) and Eq.~(\ref{F2}), we have
\begin{eqnarray}\label{F10}
\lambda\psi_{i\to j\to k}=\left\{\begin{aligned}
&x_{j\to k}-\psi_{j\to k\to i}, i\in\mathcal{N}_k \\
&x_{j\to k}, i\notin\mathcal{N}_k.
\end{aligned}
\right.
\end{eqnarray}
Thus, for the case of nodes $i$, $j$ and $k$ constituting a triangle, we can establish the following equations:
\begin{eqnarray*}
\left\{\begin{aligned}
&\lambda\psi_{i\to j\to k}+\psi_{j\to k\to i}=x_{j\to k}\\
&\lambda\psi_{j\to k\to i}+\psi_{k\to i\to j}=x_{k\to i}\\
&\lambda\psi_{k\to i\to j}+\psi_{i\to j\to k}=x_{i\to j}
\end{aligned}\right..
\end{eqnarray*}
When $\lambda^3+1\neq 0$, i.e. $\lambda\neq -1, (1+\sqrt{3}\mathrm{i})/2$ or $(1-\sqrt{3}\mathrm{i})/2$,  we have solutions as
\begin{eqnarray}\label{F30}
\left\{\begin{aligned}
&\psi_{i\to j\to k}=\frac{1}{\lambda^3+1}(\lambda^2 x_{j\to k}-\lambda x_{k\to i}+x_{i\to j})\\
&\psi_{j\to k\to i}=\frac{1}{\lambda^3+1}(\lambda^2 x_{k\to i}-\lambda x_{i\to j}+x_{j\to k})\\
&\psi_{k\to i\to j}=\frac{1}{\lambda^3+1}(\lambda^2 x_{i\to j}-\lambda x_{j\to k}+x_{k\to i})
\end{aligned}
\right..
\end{eqnarray}
Combining Eq.~\eqref{F10} and Eq.~\eqref{F30}, we have
\begin{eqnarray}\label{F31}
\psi_{i\to j\to k}=\left\{\begin{aligned}
&\frac{1}{\lambda^3+1}(\lambda^2 x_{j\to k}-\lambda x_{k\to i}+x_{i\to j}), i\in\mathcal{N}_k\\
&\frac{1}{\lambda}x_{j\to k}, i\notin\mathcal{N}_k
\end{aligned}
\right..
\end{eqnarray}
Thus, for any non-zero eigenvalue $\lambda$ of $M$ except $-1$, $(1+\sqrt{3}\mathrm{i})/2$ and $(1-\sqrt{3}\mathrm{i})/2$, it satisfies $B^{(2)}\psi=\lambda\psi$, and the values of elements in $\psi$ can be determined by Eq.~\eqref{F31}. 
\end{proof}

Exploiting Theorem~\ref{T2}, we give an approach to save time and space cost for computing $p_c^{(2)}$. 
In particular, the numbers of edges and length-2 directed paths in a network $\mathcal{G}$ are $E=\frac{1}{2}\sum_{i=1}^{N}d_i$ and $P_2=\sum_{i=1}^{N}d_i(d_i-1)$, respectively, where $d_i$ is the degree of node $i$. Thus, we reduce the size of the matrix to be computed by a factor of $(\sum_{i=1}^{N}d_i^2-\sum_{i=1}^{N}d_i)/(4\sum_{i=1}^{N}d_i)\approx \langle d^2\rangle/(4\langle d\rangle)$. 

\section{Experiments} \label{sec:exp}

In this section, we empirically investigate the validity of our theoretical estimate for bond percolation threshold.
First, we compare $p_c^{(2)}$ with other theoretical indicators $p_c^{(1)}$ and $p_c^{(0)}$ on a class of synthetic networks
	generated by the forest fire model \cite{LeKlFa07}. 
Next, we perform extensive experiments on 42 real networks to further explore the performance of these estimations.


We use the peak value of the second largest component (Eq.~\eqref{A0}) as the ground truth of the bond percolation
	threshold $p_c$.
For each network, the value of the empirical ground truth $p_c$
	 is computed by 1000 independent Monte Carlo simulations of the percolation process as proposed by Newman and Ziff~\cite{NeZi00}.

\begin{figure}[t]
	\centering
	\captionsetup[subfigure]{font=scriptsize,oneside,margin={0.6cm,0.0cm}}
	\setcounter{subfigure}{0}%
	\subfloat[\textrm{burning probability $0.01$}]
	{\includegraphics[width=0.5\linewidth,trim= 50 0 50 100]{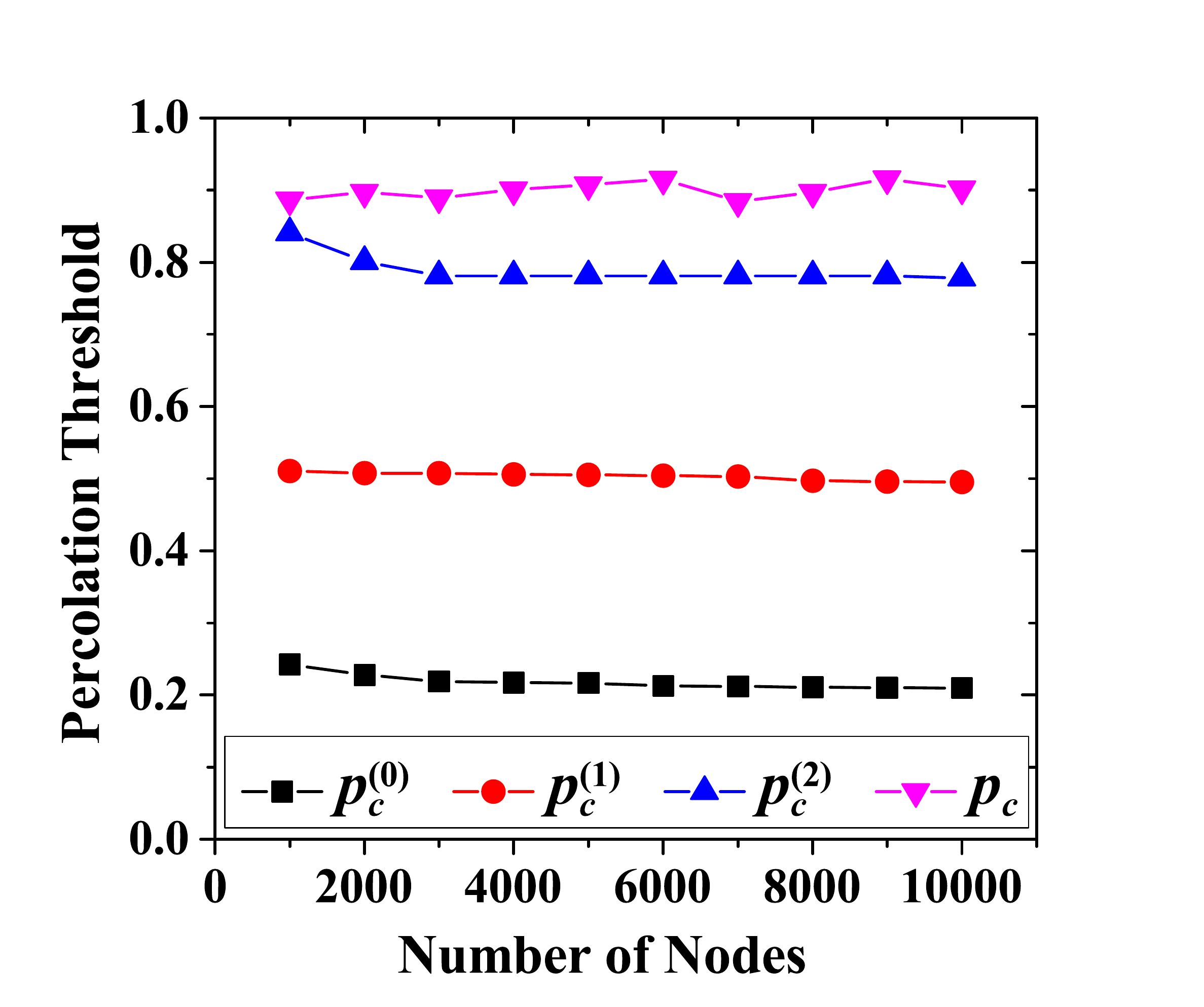}\label{fig:FFa}}
	\subfloat[\textrm{5000 nodes}]
	{\includegraphics[width=0.5\linewidth,trim= 50 0 50 100]{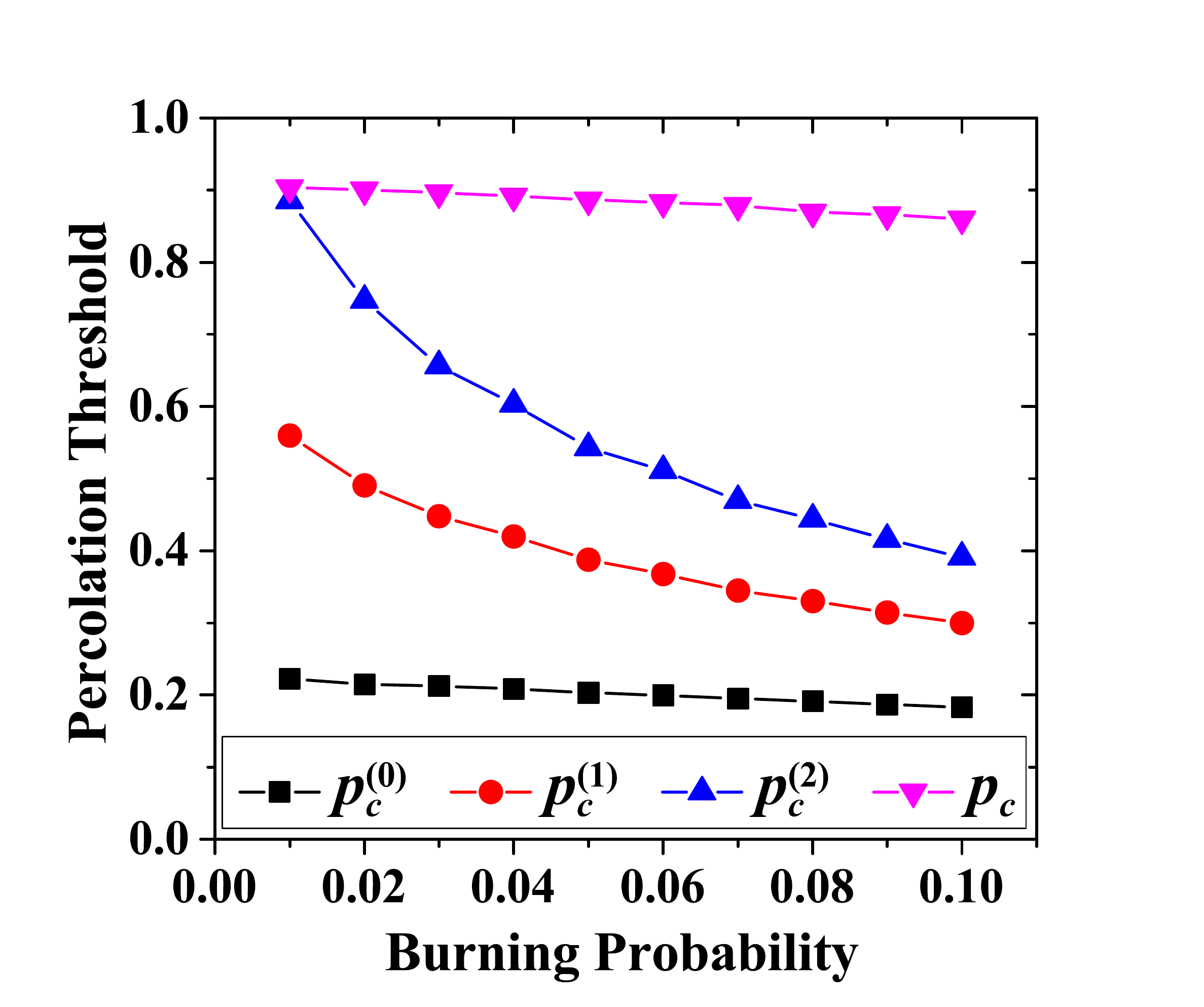}\label{fig:FFb}}
	\vspace{-0.2cm}
	\caption{Empirical and theoretical estimates of the percolation threshold on the forest fire model.}
	\label{fig:forestfire}
\end{figure}

\begin{figure}[t]
	\centering
	\captionsetup[subfigure]{font=scriptsize,oneside,margin={0.0cm,0.0cm}}
	\setcounter{subfigure}{0}%
	\subfloat[\textrm{scatter plot}]
	{\includegraphics[width=0.5\linewidth,trim= 50 0 50 100]{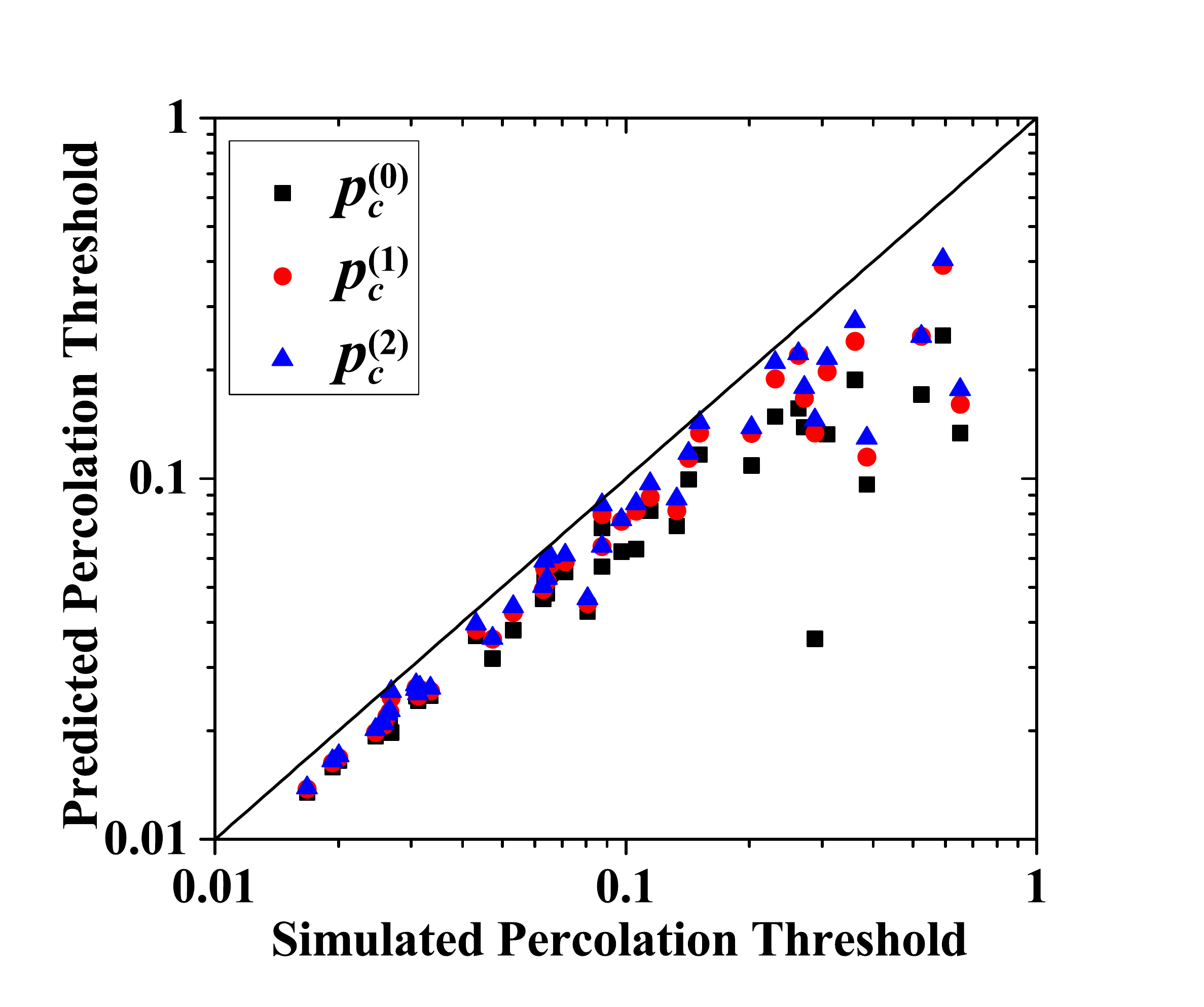}\label{fig:reala}}
	\subfloat[\textrm{cumulative distribution of relative errors}]
	{\includegraphics[width=0.5\linewidth,trim= 50 0 50 100]{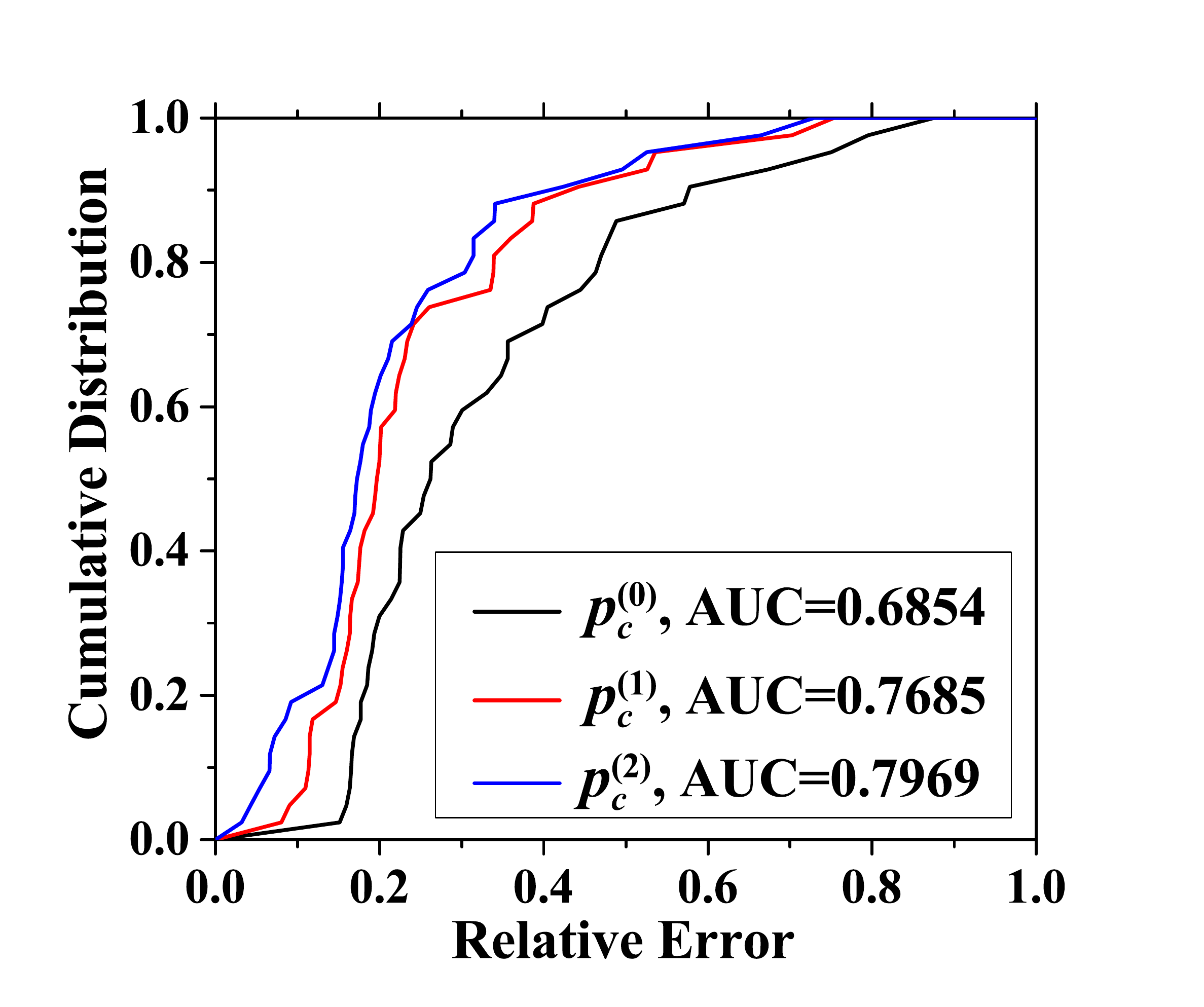}\label{fig:realb}}
	\vspace{-0.2cm}
	\caption{Test results of empirical and theoretical percolation thresholds on 42 real-world networks.}
	\label{fig:real}
\end{figure}

\begin{figure*}[t]
	\centering
	\captionsetup[subfigure]{font=scriptsize,oneside,margin={-0.5cm,-0.5cm}}
	\setcounter{subfigure}{0}%
	\subfloat[\textrm{relative errors vs. empirical estimates}]
	{\includegraphics[width=0.25\linewidth,trim= 50 0 50 100]{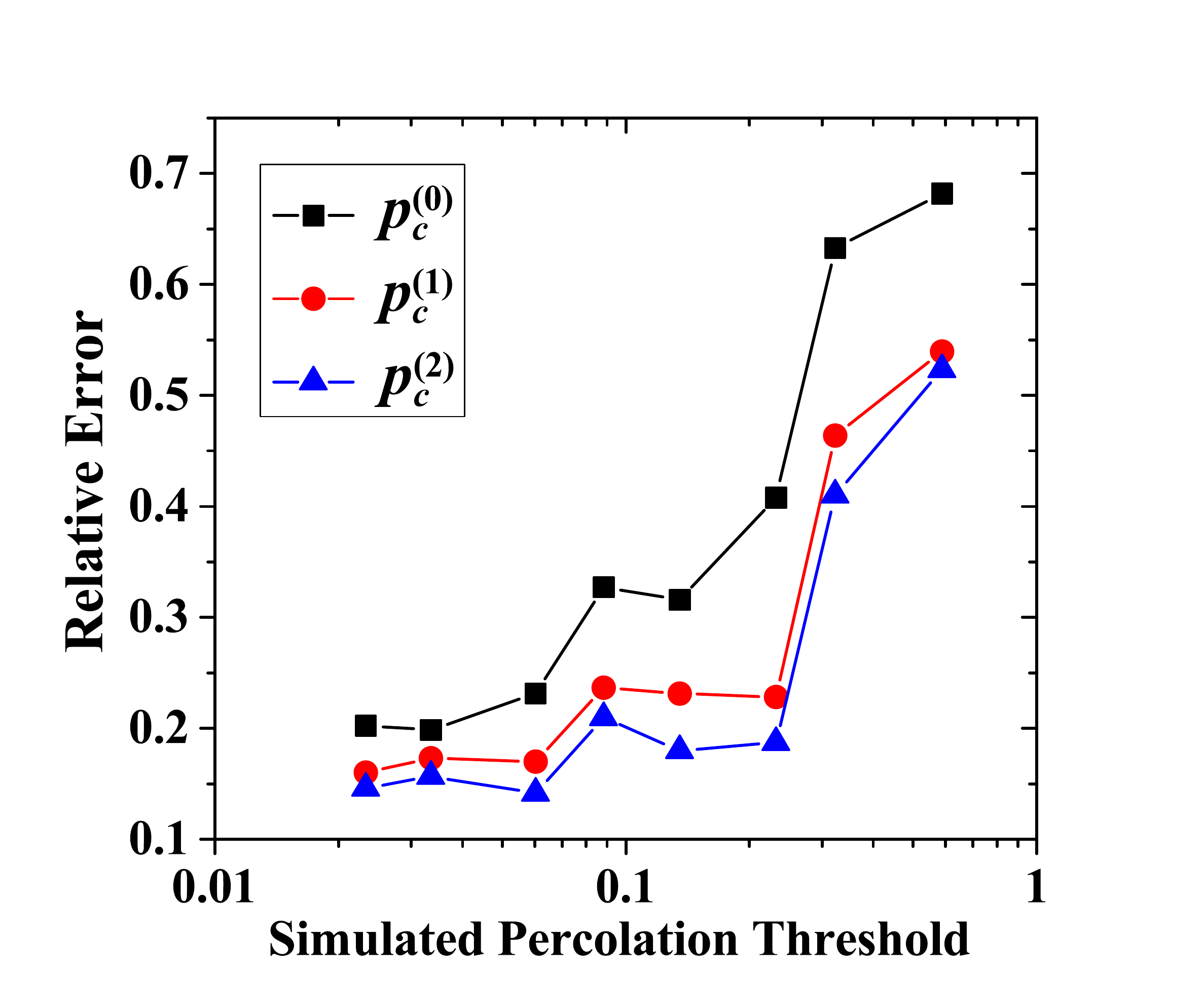}\label{fig:realc}}
	\hspace{1cm}
	\subfloat[\textrm{relative errors vs. average degrees}]
	{\includegraphics[width=0.25\linewidth,trim= 50 0 50 100]{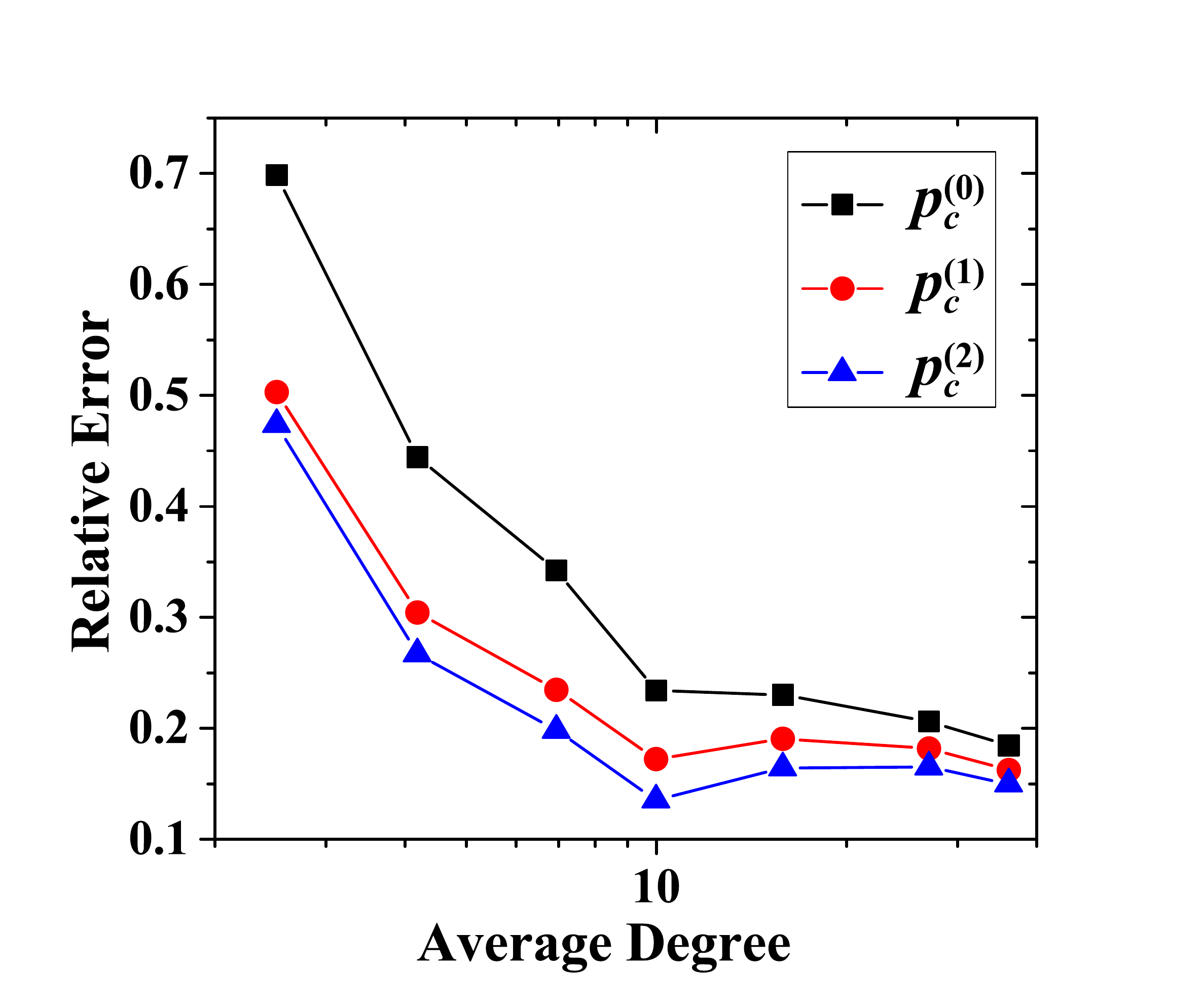}\label{fig:reald}}
	\hspace{1cm}
	\subfloat[\textrm{relative errors vs. categories}]
	{\includegraphics[width=0.25\linewidth,trim= 50 0 50 0]{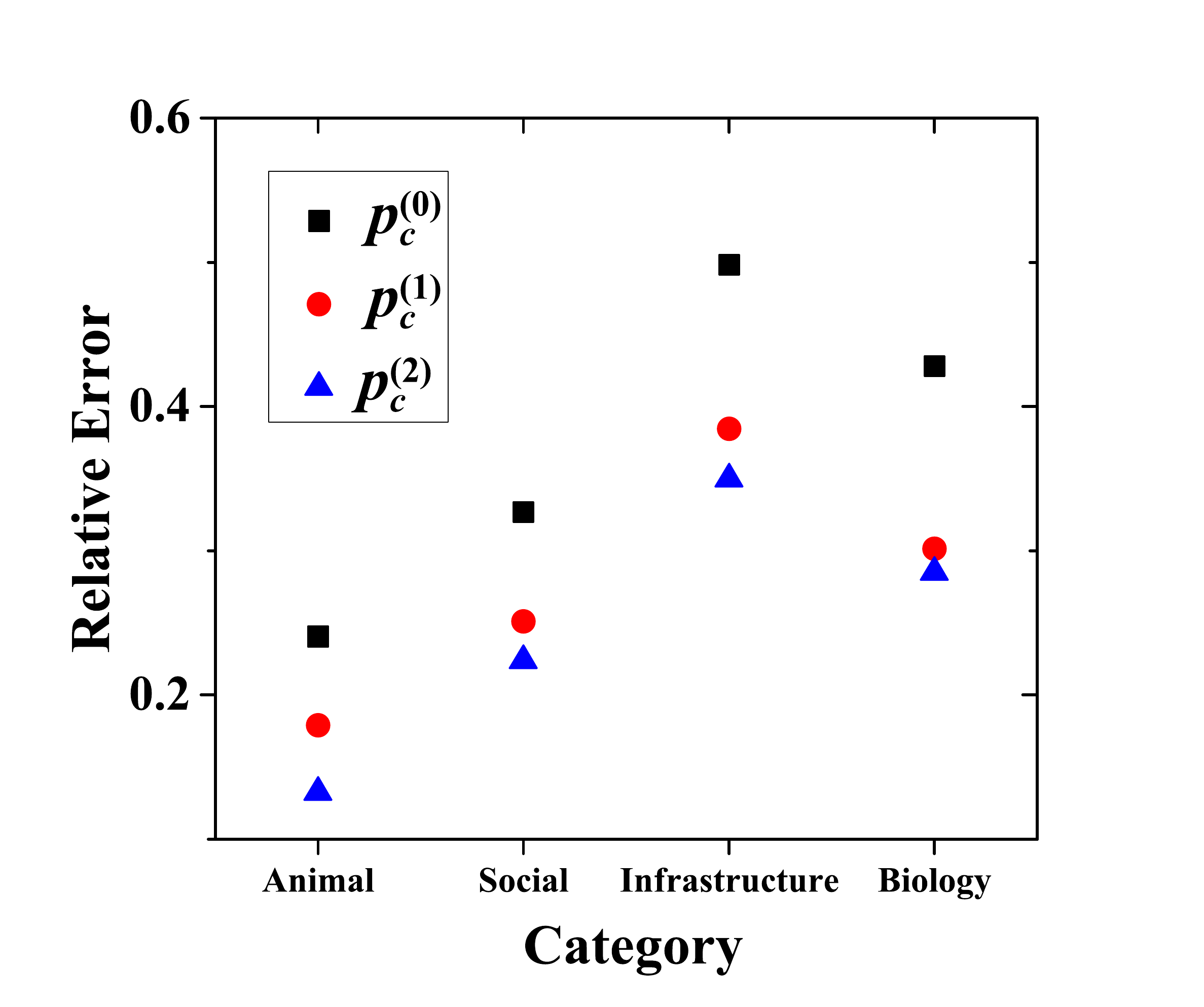}\label{fig:cate}}
	\vspace{-0.2cm}
	\caption{Relative errors of theoretical percolation thresholds on 42 real-world networks.}
	\label{fig:real2}
\end{figure*}

\subsection{Forest Fire Model}

The forest fire model \cite{LeKlFa07} is a family of evolutionary networks, controlled by burning probability $q$. 

We denote the network at time $t$ as $G_t$. Initially ($t=1$), there is only one node in $G_1$. At time $t>1$, there is a new node $u$ joining the network $G_{t-1}$ and generating $G_t$. 
The node $u$ establishes connections with other existing nodes through the following process: 
	(1) Node $u$ first selects an ambassador node $v$ in $G_{t-1}$ uniformly at random and connects to it. 
	(2) We sample a random number $a\in(0,1]$. If $a \le q$, node $u$ randomly chooses a neighbor of $v$, which is not connected to $u$ yet, and forms a link to it. Then repeat this step. If $a>q$, this step ends and we label all nodes linking to $u$ at this step as $w_1, w_2, \cdots, w_k$; 
	(3) Let $w_1, w_2, \cdots, w_k$ sequentially be the ambassador of node $u$ and apply step (2) for each of them recursively. A node in the process should not be visited a second time.

The generation process of the forest fire model describes new nodes joining social networks with an epidemic fashion. 
In addition, forest fire model shares a number of remarkable structural features with real networks, such as densification and shrinking diameters~\cite{LeKlFa07}. 
Thus, studying percolation on the forest fire model can enhance our understanding of spreading in realistic systems.

We first consider a case that the burning probability is very small, so that the network $G_t$ is sparse. 
In Figure~\ref{fig:FFa}, we plot bond percolation threshold estimated by different approaches during the evolution of a forest fire model with $q=0.01$. 
It shows that along with the growth of the network, the empirical estimate tends to $0.901$. 
The predicted percolation thresholds based on the adjacency and non-backtracking matrices 
	tend to $0.210$ and $0.495$, respectively, which are far from the empirical estimate. 
Although $p_c^{(1)}$ performs better than $p_c^{(0)}$ in sparse networks, the appearance of triangle structures largely weakens its effectiveness. 
The estimation based on the 2nd-order non-backtracking matrix shows remarkable improvement. 
At the end of the evolution, $p_c^{(2)}$ tends to $0.778$, which corresponds to an improvement of roughly $30\%$ from $p_c^{(1)}$. It indicates $p_c^{(2)}$ is evidently more precise than $p_c^{(1)}$ for sparse networks with triangles.

Next we test the impact of burning probability on our indicator.
In Figure~\ref{fig:FFb}, we show the percolation threshold given by the theoretical and simulated estimations versus burning probability on forest fire model with 5000 nodes. 
Each value is the average of 100 networks.
We observe that, along with the growth of the burning probability, the empirical percolation threshold decreases, because of the increase of the number of edges.
For the theoretical indicators, it is not surprising that $p_c^{(2)}$ always performs better than $p_c^{(1)}$ and $p_c^{(0)}$, especially in the regime of small burning probability. However, as burning probability becomes large, the precision of $p_c^{(2)}$ decreases. 
In addition, the values of $p_c^{(0)}$,$p_c^{(1)}$ and $p_c^{(2)}$ are getting close.
This is because the growth of the burning probability introduces more and more simple cycles with length larger than 3, which would lead to overestimating probability $\theta_{i\to j\to k}$ for $i\to j\to k$ in such a cycle. Thus, $p_c^{(2)}$ becomes less effective in predicting the percolation threshold.

\subsection{Real Networks}

For the purpose of better understanding the predictive power of different theoretical estimations in real networks, we evaluate bond percolation threshold predicted by $p_c^{(0)}$, $p_c^{(1)}$, $p_c^{(2)}$ and simulations on 42 real networks. The dataset includes social, infrastructural, animal and biological networks, which exhibit various topological properties, with size from 23 to 4941. In Table~\ref{Stas}, we report the theoretical and empirical estimations for bond percolation threshold on the 42 real networks. All networks are treated as undirected unweighted networks. We only consider the largest connected component in each network. All datasets are from the Koblenz Network Collection (http://konect.uni-koblenz.de/).

\OnlyInFull{In Appendix~\ref{App2}, we display the structural information and detailed results for each network.}



Figure~\ref{fig:reala} reports values of $p_c^{(0)}$, $p_c^{(1)}$ and $p_c^{(2)}$ as functions of the empirical estimate for all 42 networks. 
It shows that $p_c^{(2)}$ is always closer to the empirical estimate than $p_c^{(1)}$ and $p_c^{(0)}$, which is in agreement with our theoretical analysis. In addition, we observe that there are some networks on which $p_c^{(2)}$ can attain the empirical estimate. However, $p_c^{(2)}$ is not always close enough to the empirical estimate, especially on networks in the regime of large percolation threshold.


Next, we continue to quantitatively analyze how much improvement our proposed method attains. 
We define relative error of the theoretical estimate as $(p_c-p_c^{(\alpha)})/p_c$ ($\alpha=0, 1$ or $2$). 
Then, we plot cumulative distribution of relative errors of $p_c^{(2)}$, $p_c^{(1)}$ and $p_c^{(0)}$ in the collection of 42 networks, as shown in Figure~\ref{fig:realb}. We adopt the area under the curve (AUC) as a measurement for global performance. The results is that there is an improvement of 8\% from $p_c^{(0)}$ to $p_c^{(1)}$, and 3\% from $p_c^{(1)}$ to $p_c^{(2)}$.




In order to further understand the influence of the true percolation threshold on our theoretical prediction, we first divide the range of possible values of the empirical estimate in eight bins. Then, we calculate the average of the empirical estimate and relative errors, respectively, in each bin. Figure~\ref{fig:realc} reports the average of relative error with respect to the average of the empirical estimate. It implies that $p_c^{(2)}$ works well on the networks in the range of small true percolation threshold; while it become less predictive on networks with large true percolation threshold.



We also analyze the impact of density of networks on our theoretical prediction. Analogously, we divide the range of possible values of average degree in seven bins and separately calculate the means of average degree and relative errors in each bin. In Figure~\ref{fig:reald}, we plot the relative error as a function of average degree for all three theoretical estimates. 
There are two main conclusions from this figure. First, for all three theoretical predictions, they are more precise on networks with large average degree than on networks with small average degree. 
Second, on networks in the regime of small average degree, $p_c^{(2)}$ in general
	has a better improvement from $p_c^{(1)}$ than on networks with relatively large average degree. 
It can be accounted as on dense networks, $p_c^{(0)}$ is already a good approximation for percolation threshold, and under such case $p_c^{(0)}$, $p_c^{(1)}$ and $p_c^{(2)}$ would be very close to each other.

At last, we discover the performance of theoretical predictions in different categories of networks. In Figure~\ref{fig:cate}, we give the average of relative errors in each group of networks for all three theoretical estimations. It can be found that theoretical predictions are closer to true percolation threshold in animal and social networks than in infrastructural and biological networks.

\begin{table}
\centering
\resizebox{\columnwidth}{!}{
\begin{threeparttable}
\caption{Percolation threshold for 42 real networks.}\label{Stas}
\footnotesize
\begin{tabular}{|c|c|c|c|c|c|c|}
\hline
\raisebox{-0.5ex}{Network} & \raisebox{-0.5ex}{Nodes} & \raisebox{-0.5ex}{Edges} & \raisebox{-0.5ex}{$p_c^{(0)}$} & \raisebox{-0.5ex}{$p_c^{(1)}$} & \raisebox{-0.5ex}{$p_c^{(2)}$}& \raisebox{-0.5ex}{$p_c$} \\
		\hline
		\hline
		\multicolumn{7}{|c|}{\raisebox{-0.5ex}{Animal Networks}} \\[0.5ex]
		\hline
		\raisebox{-0.5ex}{Zebra} & \raisebox{-0.5ex}{23} & \raisebox{-0.5ex}{105} & \raisebox{-0.5ex}{0.0814} & \raisebox{-0.5ex}{0.0889} & \raisebox{-0.5ex}{0.0969} & \raisebox{-0.5ex}{0.1146} \\[0.5ex]
		\hline
		\raisebox{-0.5ex}{Bison} & \raisebox{-0.5ex}{26} & \raisebox{-0.5ex}{222} & \raisebox{-0.5ex}{0.0544} & \raisebox{-0.5ex}{0.0577} & \raisebox{-0.5ex}{0.0608} & \raisebox{-0.5ex}{0.0651} \\[0.5ex]
		\hline
		\raisebox{-0.5ex}{Cattle} & \raisebox{-0.5ex}{28} & \raisebox{-0.5ex}{160} & \raisebox{-0.5ex}{0.0729} & \raisebox{-0.5ex}{0.0796} & \raisebox{-0.5ex}{0.0847} & \raisebox{-0.5ex}{0.0875} \\[0.5ex]
		\hline
		\raisebox{-0.5ex}{Sheep} & \raisebox{-0.5ex}{28} & \raisebox{-0.5ex}{235} & \raisebox{-0.5ex}{0.0548} & \raisebox{-0.5ex}{0.0581} & \raisebox{-0.5ex}{0.0609} & \raisebox{-0.5ex}{0.0656} \\[0.5ex]
		\hline
		\raisebox{-0.5ex}{Dolphins} & \raisebox{-0.5ex}{62} & \raisebox{-0.5ex}{159} & \raisebox{-0.5ex}{0.1390} & \raisebox{-0.5ex}{0.1668} & \raisebox{-0.5ex}{0.1791} & \raisebox{-0.5ex}{0.2717} \\[0.5ex]
		\hline
		\raisebox{-0.5ex}{Macaques} & \raisebox{-0.5ex}{62} & \raisebox{-0.5ex}{1167} & \raisebox{-0.5ex}{0.0256} & \raisebox{-0.5ex}{0.0263} & \raisebox{-0.5ex}{0.0268} & \raisebox{-0.5ex}{0.0308} \\[0.5ex]
		\hline
		\multicolumn{7}{|c|}{\raisebox{-0.5ex}{Social Networks}} \\[0.5ex]
		\hline
		\raisebox{-0.5ex}{Seventh Graders} & \raisebox{-0.5ex}{29} & \raisebox{-0.5ex}{250} & \raisebox{-0.5ex}{0.0532} & \raisebox{-0.5ex}{0.0564} & \raisebox{-0.5ex}{0.0592} & \raisebox{-0.5ex}{0.0633} \\[0.5ex]
		\hline
		\raisebox{-0.5ex}{Dutch College} & \raisebox{-0.5ex}{32} & \raisebox{-0.5ex}{422} & \raisebox{-0.5ex}{0.0367} & \raisebox{-0.5ex}{0.0381} & \raisebox{-0.5ex}{0.0395} & \raisebox{-0.5ex}{0.0432} \\[0.5ex]
		\hline
		\raisebox{-0.5ex}{Zachary Karate Club} & \raisebox{-0.5ex}{34} & \raisebox{-0.5ex}{78} & \raisebox{-0.5ex}{0.1487} & \raisebox{-0.5ex}{0.1889} & \raisebox{-0.5ex}{0.2097} & \raisebox{-0.5ex}{0.2310} \\[0.5ex]
		\hline
		\raisebox{-0.5ex}{Windsurfers} & \raisebox{-0.5ex}{43} & \raisebox{-0.5ex}{336} & \raisebox{-0.5ex}{0.0552} & \raisebox{-0.5ex}{0.0588} & \raisebox{-0.5ex}{0.0614} & \raisebox{-0.5ex}{0.0712} \\[0.5ex]
		\hline
		\raisebox{-0.5ex}{Train bombing} & \raisebox{-0.5ex}{64} & \raisebox{-0.5ex}{243} & \raisebox{-0.5ex}{0.0739} & \raisebox{-0.5ex}{0.0815} & \raisebox{-0.5ex}{0.0878} & \raisebox{-0.5ex}{0.1330} \\[0.5ex]
		\hline
		\raisebox{-0.5ex}{Hypertext 2009} & \raisebox{-0.5ex}{113} & \raisebox{-0.5ex}{2196} & \raisebox{-0.5ex}{0.0214} & \raisebox{-0.5ex}{0.0219} & \raisebox{-0.5ex}{0.0222} & \raisebox{-0.5ex}{0.0262} \\[0.5ex]
		\hline
		\raisebox{-0.5ex}{Physicians} & \raisebox{-0.5ex}{117} & \raisebox{-0.5ex}{465} & \raisebox{-0.5ex}{0.0994} & \raisebox{-0.5ex}{0.1138} & \raisebox{-0.5ex}{0.1176} & \raisebox{-0.5ex}{0.1421} \\[0.5ex]
		\hline
		\raisebox{-0.5ex}{Manufacturing Emails} & \raisebox{-0.5ex}{167} & \raisebox{-0.5ex}{3250} & \raisebox{-0.5ex}{0.0165} & \raisebox{-0.5ex}{0.0168} & \raisebox{-0.5ex}{0.0171} & \raisebox{-0.5ex}{0.0200} \\[0.5ex]
		\hline
		\raisebox{-0.5ex}{Jazz Musicians} & \raisebox{-0.5ex}{198} & \raisebox{-0.5ex}{2742} & \raisebox{-0.5ex}{0.0250} & \raisebox{-0.5ex}{0.0258} & \raisebox{-0.5ex}{0.0263} & \raisebox{-0.5ex}{0.0335} \\[0.5ex]
		\hline
		\raisebox{-0.5ex}{Residence Hall} & \raisebox{-0.5ex}{217} & \raisebox{-0.5ex}{1839} & \raisebox{-0.5ex}{0.0464} & \raisebox{-0.5ex}{0.0492} & \raisebox{-0.5ex}{0.0503} & \raisebox{-0.5ex}{0.0629} \\[0.5ex]
		\hline
		\raisebox{-0.5ex}{Haggle} & \raisebox{-0.5ex}{274} & \raisebox{-0.5ex}{2124} & \raisebox{-0.5ex}{0.0193} & \raisebox{-0.5ex}{0.0198} & \raisebox{-0.5ex}{0.0202} & \raisebox{-0.5ex}{0.0246} \\[0.5ex]
		\hline
		\raisebox{-0.5ex}{Network Science} & \raisebox{-0.5ex}{379} & \raisebox{-0.5ex}{914} & \raisebox{-0.5ex}{0.0964} & \raisebox{-0.5ex}{0.1148} & \raisebox{-0.5ex}{0.1294} & \raisebox{-0.5ex}{0.3862} \\[0.5ex]
		\hline
		\raisebox{-0.5ex}{Infectious} & \raisebox{-0.5ex}{410} & \raisebox{-0.5ex}{2765} & \raisebox{-0.5ex}{0.0428} & \raisebox{-0.5ex}{0.0450} & \raisebox{-0.5ex}{0.0465} & \raisebox{-0.5ex}{0.0807} \\[0.5ex]
		\hline
		\raisebox{-0.5ex}{Crime} & \raisebox{-0.5ex}{829} & \raisebox{-0.5ex}{1473} & \raisebox{-0.5ex}{0.1565} & \raisebox{-0.5ex}{0.2198} & \raisebox{-0.5ex}{0.2220} & \raisebox{-0.5ex}{0.2628} \\[0.5ex]
		\hline
		\raisebox{-0.5ex}{Email} & \raisebox{-0.5ex}{1133} & \raisebox{-0.5ex}{5451} & \raisebox{-0.5ex}{0.0482} & \raisebox{-0.5ex}{0.0519} & \raisebox{-0.5ex}{0.0529} & \raisebox{-0.5ex}{0.0642} \\[0.5ex]
		\hline
		\raisebox{-0.5ex}{Hamsterster Friendships} & \raisebox{-0.5ex}{1788} & \raisebox{-0.5ex}{12476} & \raisebox{-0.5ex}{0.0217} & \raisebox{-0.5ex}{0.0226} & \raisebox{-0.5ex}{0.0228} & \raisebox{-0.5ex}{0.0266} \\[0.5ex]
		\hline
		\raisebox{-0.5ex}{Hamsterster Full} & \raisebox{-0.5ex}{2000} & \raisebox{-0.5ex}{16098} & \raisebox{-0.5ex}{0.0200} & \raisebox{-0.5ex}{0.0207} & \raisebox{-0.5ex}{0.0210} & \raisebox{-0.5ex}{0.0258} \\[0.5ex]
		\hline
		\raisebox{-0.5ex}{Facebook} & \raisebox{-0.5ex}{2888} & \raisebox{-0.5ex}{2981} & \raisebox{-0.5ex}{0.0360} & \raisebox{-0.5ex}{0.1339} & \raisebox{-0.5ex}{0.1456} & \raisebox{-0.5ex}{0.2885} \\[0.5ex]
		\hline
		\multicolumn{7}{|c|}{\raisebox{-0.5ex}{Infrastructural Networks}} \\[0.5ex]
		\hline
		\raisebox{-0.5ex}{Contiguous USA} & \raisebox{-0.5ex}{49} & \raisebox{-0.5ex}{107} & \raisebox{-0.5ex}{0.1880} & \raisebox{-0.5ex}{0.2400} & \raisebox{-0.5ex}{0.2723} & \raisebox{-0.5ex}{0.3610} \\[0.5ex]
		\hline
		\raisebox{-0.5ex}{Euroroad} & \raisebox{-0.5ex}{1039} & \raisebox{-0.5ex}{1305} & \raisebox{-0.5ex}{0.2493} & \raisebox{-0.5ex}{0.3906} & \raisebox{-0.5ex}{0.4050} & \raisebox{-0.5ex}{0.5908} \\[0.5ex]
		\hline
		\raisebox{-0.5ex}{Air Traffic} & \raisebox{-0.5ex}{1226} & \raisebox{-0.5ex}{2408} & \raisebox{-0.5ex}{0.1086} & \raisebox{-0.5ex}{0.1336} & \raisebox{-0.5ex}{0.1387} & \raisebox{-0.5ex}{0.2022} \\[0.5ex]
		\hline
		\raisebox{-0.5ex}{Open Flights} & \raisebox{-0.5ex}{2905} & \raisebox{-0.5ex}{15645} & \raisebox{-0.5ex}{0.0159} & \raisebox{-0.5ex}{0.0163} & \raisebox{-0.5ex}{0.0165} & \raisebox{-0.5ex}{0.0193} \\[0.5ex]
		\hline
		\raisebox{-0.5ex}{US Power} & \raisebox{-0.5ex}{4941} & \raisebox{-0.5ex}{6594} & \raisebox{-0.5ex}{0.1336} & \raisebox{-0.5ex}{0.1606} & \raisebox{-0.5ex}{0.1766} & \raisebox{-0.5ex}{0.6518} \\[0.5ex]
		\hline
		\multicolumn{7}{|c|}{\raisebox{-0.5ex}{Biological Networks}} \\[0.5ex]
		\hline
		\raisebox{-0.5ex}{PDZBase} & \raisebox{-0.5ex}{161} & \raisebox{-0.5ex}{209} & \raisebox{-0.5ex}{0.1712} & \raisebox{-0.5ex}{0.2483} & \raisebox{-0.5ex}{0.2484} & \raisebox{-0.5ex}{0.5237} \\[0.5ex]
		\hline
		\raisebox{-0.5ex}{Caenorhabditis Elegans} & \raisebox{-0.5ex}{453} & \raisebox{-0.5ex}{2025} & \raisebox{-0.5ex}{0.0380} & \raisebox{-0.5ex}{0.0426} & \raisebox{-0.5ex}{0.0442} & \raisebox{-0.5ex}{0.0533} \\[0.5ex]
		\hline
		\raisebox{-0.5ex}{Protein} & \raisebox{-0.5ex}{1458} & \raisebox{-0.5ex}{1948} & \raisebox{-0.5ex}{0.1327} & \raisebox{-0.5ex}{0.1980} & \raisebox{-0.5ex}{0.2154} & \raisebox{-0.5ex}{0.3094} \\[0.5ex]
		\hline
		\raisebox{-0.5ex}{Human Protein} & \raisebox{-0.5ex}{1615} & \raisebox{-0.5ex}{3106} & \raisebox{-0.5ex}{0.0571} & \raisebox{-0.5ex}{0.0648} & \raisebox{-0.5ex}{0.0649} & \raisebox{-0.5ex}{0.0876} \\[0.5ex]
		\hline
		\raisebox{-0.5ex}{Protein Figeys} & \raisebox{-0.5ex}{2217} & \raisebox{-0.5ex}{6418} & \raisebox{-0.5ex}{0.0317} & \raisebox{-0.5ex}{0.0359} & \raisebox{-0.5ex}{0.0360} & \raisebox{-0.5ex}{0.0474} \\[0.5ex]
		\hline
		\raisebox{-0.5ex}{Protein Vidal} & \raisebox{-0.5ex}{2783} & \raisebox{-0.5ex}{6007} & \raisebox{-0.5ex}{0.0628} & \raisebox{-0.5ex}{0.0761} & \raisebox{-0.5ex}{0.0770} & \raisebox{-0.5ex}{0.0975} \\[0.5ex]
		\hline
		\multicolumn{7}{|c|}{\raisebox{-0.5ex}{Other Networks}} \\[0.5ex]
		\hline
		\raisebox{-0.5ex}{Corporate Leadership} & \raisebox{-0.5ex}{24} & \raisebox{-0.5ex}{86} & \raisebox{-0.5ex}{0.1167} & \raisebox{-0.5ex}{0.1339} & \raisebox{-0.5ex}{0.1430} & \raisebox{-0.5ex}{0.1512} \\[0.5ex]
		\hline
		\raisebox{-0.5ex}{Florida Ecosystem Dry} & \raisebox{-0.5ex}{128} & \raisebox{-0.5ex}{2106} & \raisebox{-0.5ex}{0.0249} & \raisebox{-0.5ex}{0.0257} & \raisebox{-0.5ex}{0.0260} & \raisebox{-0.5ex}{0.0308} \\[0.5ex]
		\hline
		\raisebox{-0.5ex}{Florida ecosystem Wet} & \raisebox{-0.5ex}{128} & \raisebox{-0.5ex}{2075} & \raisebox{-0.5ex}{0.0252} & \raisebox{-0.5ex}{0.0260} & \raisebox{-0.5ex}{0.0264} & \raisebox{-0.5ex}{0.0315} \\[0.5ex]
		\hline
		\raisebox{-0.5ex}{Little Rock Lake} & \raisebox{-0.5ex}{183} & \raisebox{-0.5ex}{2434} & \raisebox{-0.5ex}{0.0242} & \raisebox{-0.5ex}{0.0250} & \raisebox{-0.5ex}{0.0253} & \raisebox{-0.5ex}{0.0313} \\[0.5ex]
		\hline
		\raisebox{-0.5ex}{Unicode Languages} & \raisebox{-0.5ex}{614} & \raisebox{-0.5ex}{1245} & \raisebox{-0.5ex}{0.0637} & \raisebox{-0.5ex}{0.0811} & \raisebox{-0.5ex}{0.0853} & \raisebox{-0.5ex}{0.1059} \\[0.5ex]
		\hline
		\raisebox{-0.5ex}{Blogs} & \raisebox{-0.5ex}{1222} & \raisebox{-0.5ex}{16714} & \raisebox{-0.5ex}{0.0135} & \raisebox{-0.5ex}{0.0138} & \raisebox{-0.5ex}{0.0139} & \raisebox{-0.5ex}{0.0167} \\[0.5ex]
		\hline
		\raisebox{-0.5ex}{Bible} & \raisebox{-0.5ex}{1773} & \raisebox{-0.5ex}{7105} & \raisebox{-0.5ex}{0.0198} & \raisebox{-0.5ex}{0.0247} & \raisebox{-0.5ex}{0.0257} & \raisebox{-0.5ex}{0.0268} \\[0.5ex]
		\hline	
	\end{tabular}
	
	\end{threeparttable}}
\end{table}

\section{Discussion and Future Work}

\begin{figure}[t]
	\begin{center}
		\includegraphics[width=0.7 \linewidth,trim= 50 70 0 0]{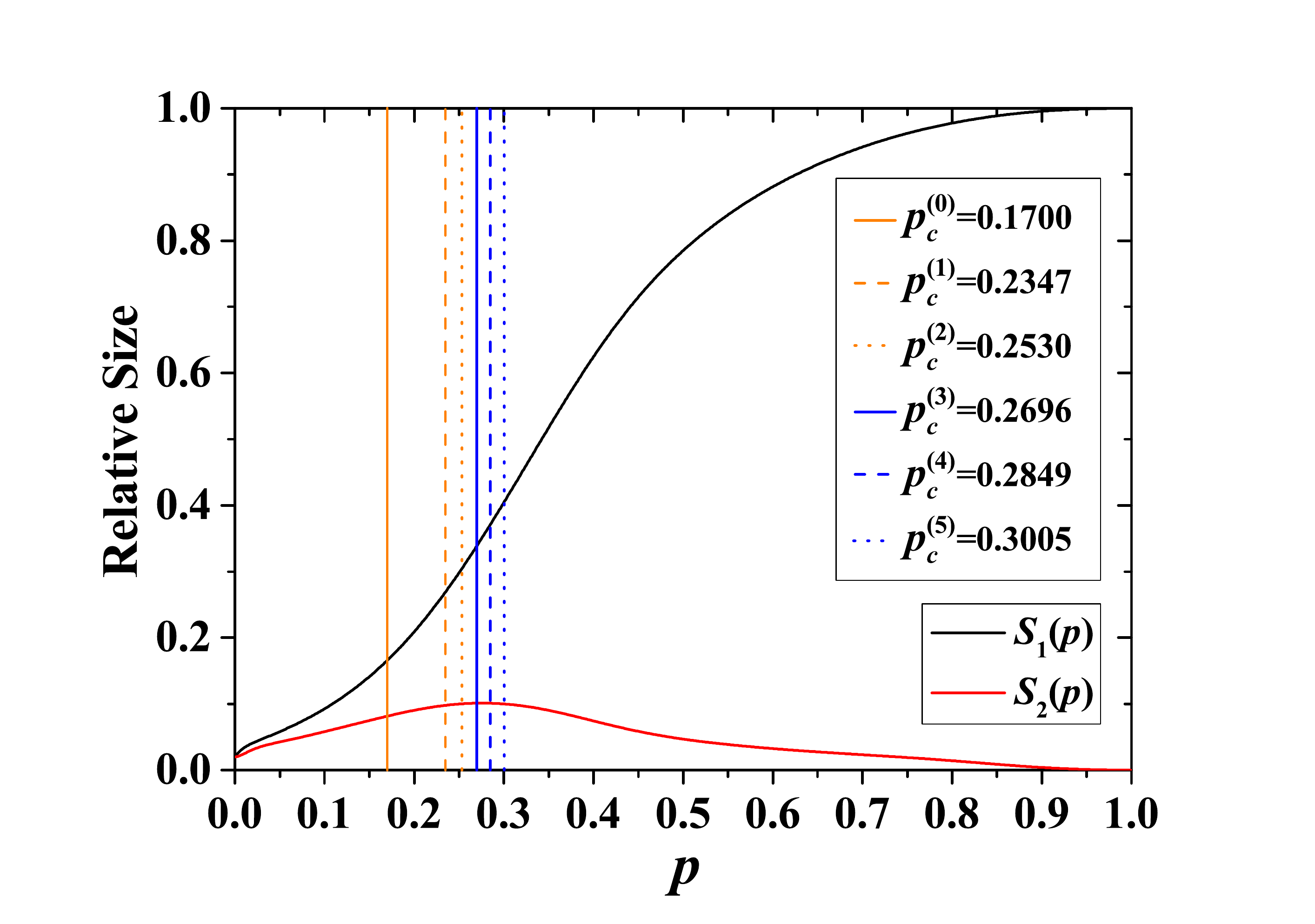}
	\end{center}
	\caption[kurzform]{Bond percolation on a Barab\'asi-Albert network with 50 nodes and mean degree 4. The empirical estimate is $p_c=0.2834$.}\label{high}{}
\end{figure}

Although $p_c^{(2)}$ outperforms other theoretical estimations for bond percolation threshold, it is still not sufficiently precise on some real networks, due to the appearance of cycles larger than 3.  
A direct idea to further improve the precision is to use the 
	reciprocal of the largest eigenvalue of higher order non-backtracking matrices as the predictor for the
	percolation threshold.
According to a similar deduction as in Section~\ref{sec:lowerbound}, applying a $g$-th-order non-backtracking matrix should avoid overestimating probability that node $i$ connecting to the giant component following a length-$g$ path. 
Moreover, Theorem~\ref{T0} also theoretically guarantees percolation threshold estimated by $g$-th-order non-backtracking matrix is non-decreasing with respect to $g$. 
In order to elaborate it empirically, we consider bond percolation on a small network with only 50 nodes, generated by the Barab\'asi-Albert network model~\cite{BaAl99}. In Figure~\ref{high}, we give the expected relative sizes of the first and second largest cluster as a function of probability $p$. In addition, we derive theoretical estimations for percolation threshold based on $g$-th-order non-backtracking matrix ($0\leq g\leq 5$), denoted separately as $p_c^{(g)}$. We can observe $p_c^{(3)}$ is a better approximation than $p_c^{(2)}$.

However, theoretical estimation based on high-order non-backtracking matrices still suffers the following two drawbacks. First, along with the increase of $g$, the size of the $g$-th-order non-backtracking matrix grows dramatically. 
It would lead to a prohibitive cost for computing the largest eigenvalue for large-scale networks. 
Second, with a sufficiently large $g$, $p_c^{(g)}$ could be larger than the true percolation threshold. 
As shown in Figure~\ref{high}, $p_c^{(4)}$ and $p_c^{(5)}$ are already slightly larger than the peak position of $S_2(p)$. It may be because the heuristic analysis in Section~\ref{sec:lowerbound} will be no longer hold for high-order non-backtracking matrices.
%
Thus, possible directions of future work include finding efficient techniques for generating and analyzing high-order non-backtracking matrices, and determining at which order $p_c^{(g)}$ is the closest to 
	but still smaller than the true percolation threshold.
It is also possible that we can depart from non-backtracking matrices and use some brand-new
	techniques to develop the theoretical assessment of the percolation threshold.

In addition, the current lower bound derivation is based on heuristic equations.
A more mathematically sound technique to show the effectiveness of high-order 
	non-backtracking matrices and the exact conditions that support this effectiveness
	is also a future direction.

\section{Acknowledgments}

The authors thank anonymous reviewers for their valuable comments and helpful suggestions. 
Yuan Lin and Zhongzhi Zhang are supported by the National Natural Science Foundation of China
	(Grant No. 11275049).
Wei Chen is partially supported by the National Natural Science Foundation of China (Grant
No. 61433014).




\bibliographystyle{abbrv}
\bibliography{percolation}  


\end{document}